\documentclass[10pt, twocolumn]{article}
\usepackage[letterpaper, right=0.75in, left=0.75in, top=1in, bottom=1in]{geometry}
\usepackage{setspace}
\usepackage[affil-it]{authblk}
\usepackage{abstract}

\usepackage{titlesec}
\titleformat*{\section}{\large\bfseries}
\titleformat*{\subsection}{\normalsize\bfseries}

\usepackage[ruled,vlined,linesnumbered,nosemicolon]{algorithm2e} 
\SetKwInput{KwInit}{Initialize}

\usepackage{pifont}
\newcommand{\cmark}{\ding{51}}%
\newcommand{\xmark}{\ding{55}}%

\usepackage{amsmath}
\usepackage{amssymb}
\usepackage{bm}
\usepackage{enumerate}
\mathchardef\mhyphen="2D 
\usepackage[titletoc,toc,title]{appendix}
\usepackage{footnote}
\usepackage{graphicx}
\usepackage{booktabs} 
\usepackage{enumitem}

\usepackage[font=small,labelfont=bf]{caption}
\usepackage[sort&compress]{natbib}

\usepackage{float}
\usepackage{url}
\usepackage[bottom]{footmisc} 
\usepackage{amsthm}
\newtheorem{theorem}{Theorem}[section]

\newtheorem{corollary}[theorem]{Corollary}
\newtheorem{proposition}[theorem]{Proposition}
\theoremstyle{plain}

\theoremstyle{definition}

\usepackage{wrapfig}
\usepackage[normalem]{ulem}

\usepackage{multirow}

\usepackage{mleftright}

\usepackage[dvipsnames]{xcolor}
\definecolor{revBlue}{RGB}{10,70,225}
\definecolor{dodgerblue}{RGB}{20, 80, 200}
\definecolor{salmon}{RGB}{240, 75, 75}
\definecolor{forestgreen}{RGB}{34, 139, 34}
\definecolor{goldenrod}{RGB}{220, 160, 30}
\definecolor{violet}{RGB}{185, 85, 210}
\definecolor{dimgray}{RGB}{105, 105, 105}
\definecolor{darkblue}{RGB}{0, 0, 250}

\DeclareMathOperator*{\argmin}{argmin}
\DeclareMathOperator*{\argmax}{argmax}

\makeatletter
\newcounter{manualsubequation}
\renewcommand{\themanualsubequation}{\alph{manualsubequation}}
\newcommand{\startsubequation}{%
  \setcounter{manualsubequation}{0}%
  \refstepcounter{equation}\ltx@label{manualsubeq\theequation}%
  \xdef\labelfor@subeq{manualsubeq\theequation}%
}
\newcommand{\tagsubequation}{%
  \stepcounter{manualsubequation}%
  \tag{\ref{\labelfor@subeq}\themanualsubequation}%
}
\let\subequationlabel\ltx@label
\makeatother

\usepackage{etoolbox}

\makeatletter
\patchcmd{\endalign}{\restorealignstate@}{\global\let\df@label\@empty\restorealignstate@}{}{}
\makeatother

\usepackage{tikz}
\usetikzlibrary{shapes, arrows.meta, positioning, automata, decorations.pathreplacing}

\usepackage{hyperref}
\hypersetup{
    colorlinks=true,
    linkcolor=black,
    filecolor=magenta,      
    urlcolor=revBlue,
    citecolor=black,
}

\begin{document}

\twocolumn[{%
  \begin{@twocolumnfalse}
    
    \title{Distributionally Robust End-to-End Portfolio Construction}	
    \author{Giorgio Costa\thanks{Email: gc2958@columbia.edu}\hspace*{0.25em} and Garud N. Iyengar\thanks{Email: garud@ieor.columbia.edu}}
    
    \affil{Department of Industrial Engineering and Operations Research,\\ Columbia University, New York, NY}
    
    \date{\today}
    
    \maketitle
    
    \begin{abstract}
      \noindent We propose an end-to-end distributionally robust
      system for portfolio construction that integrates the
      asset return prediction model with a distributionally robust portfolio
      optimization model. We also show how to learn
      the risk-tolerance parameter and the degree of robustness
      directly from data. End-to-end systems have an advantage
      in that information can be communicated between the
      prediction and decision layers during training, allowing
      the parameters to be trained for the final task rather
      than solely for predictive performance. However, existing
      end-to-end systems are not able to quantify and correct
      for the impact of model risk on the decision layer. Our
      proposed distributionally robust end-to-end portfolio
      selection system explicitly accounts for the impact
      of model risk. The decision layer chooses portfolios by
      solving a minimax problem where the distribution of the
      asset returns is assumed to belong to an ambiguity set
      centered around a nominal distribution. Using convex
      duality, we recast the minimax problem in a form that
      allows for efficient training of the end-to-end system. 
    \end{abstract}
    
    \bigskip
    
  \end{@twocolumnfalse}
}]
\saythanks

\section{Introduction}\label{sec:intro}

Quantitative asset management methods typically predict the 
distribution of future asset returns by a parametric model, which is then 
used as input by the decision models that construct the portfolio of asset holdings. 
This two-stage ``predict-then-optimize'' design, though intuitively appealing, is
effective only under ideal conditions, i.e. when market conditions are stationary 
and there is sufficient data, the resulting portfolios have good performance. However, 
in practice, predictions are often \emph{unreliable}, and consequently, 
the resulting portfolios have poor out-of-sample performance~\citep{chopra1993, 
merton1980estimating, best1991sensitivity, broadie1993computing}. There are two 
mitigation strategies: add a measure of risk to control variability, and allow for 
model robustness. However, these approaches involve unknown parameters 
that are hard to set in practice~\citep{bertsimas2018data}. 

In decision-making systems, we are concerned with
\emph{decision} errors instead of just \emph{prediction} errors. Although
it is well known that splitting the prediction and optimization task is
not optimal for mitigating decision errors, it is only recently that the
two steps have been combined into a ``smart-predict-then-optimize''
framework where parameters are chosen to optimize the performance 
of the corresponding decision~\citep{elmachtoub2020decision}. 
This framework has been implemented as an end-to-end system where one 
goes directly from data to decision by combining the prediction and 
decision layers, and back-propagating the gradient information 
through both decision and prediction layers during 
training~\citep{donti2017task}. In the context of portfolio construction,
such an end-to-end system allows us to learn the parameters of a
\emph{given} parametric prediction model, improving the
performance of a \emph{given} fixed portfolio selection problem. 
However, these end-to-end systems cannot accommodate robust or 
distributionally robust decision layers that provide robustness with 
respect to the prediction model, nor can they accommodate optimization 
layers with learnable parameters.  

We show that the end-to-end approach can be successfully extended to 
settings where the decision is chosen by solving a distributionally robust 
optimization problem. Introducing robustness regularizes the decision and 
improves its out-of-sample performance. This is a natural next step in the 
evolution of this approach. Although we use portfolio selection as a test 
case for the robust end-to-end framework since predictions, decisions and 
uncertainty are especially important in this problem, it will be clear from 
our model development that the approach itself can be applied to any robust 
decision problem.

\begin{figure*}[!ht]
\centering
\begin{tikzpicture}[%
    >=stealth,
    node distance=5em,
    on grid,
    auto
  ]
  	\node[draw, rectangle, minimum width=5.5em, minimum height=4em] (A) {
  	$\begin{matrix}\{\bm{x}\}_{j=t-T}^{t}\\[1.5ex] \{\bm{y}\}_{j=t-T}^{t+v} \end{matrix}$
  	};
    \node[align=center, anchor=south] (lab) at (A.north) {Inputs};
	\draw[decorate, line width=0.1em, text width=2.5em, midway, align=center, decoration={brace,mirror,raise=0.5ex}] (A.north west) -- (A.south west)  node[midway,left=1ex]{Fwd pass} ;
	\node[draw, rectangle, minimum width=10.5em, minimum height=4em, right=of A] (B) [below right=-2em and 0.75em of A.east] {
	$
	\begin{aligned}
		\hat{\bm{y}} &= \{g_{\bm{\theta}}(\bm{x}_j)\}_{j=t-T}^t\\[1ex]
		\color{dodgerblue}\bm{\epsilon} &\color{dodgerblue}= \{\bm{y}_j - \hat{\bm{y}}_j\}_{j=t-T}^{t-1}
	\end{aligned}
	$
	};
    \node[align=center, anchor=south] (lab) at (B.north) {Prediction layer};
	\draw [->,thick] (A) -- (B);
	\node[draw, rectangle, minimum width=18em, minimum height=4em, right= of B] (C) [below right=-2em and 0.75em of B.east] {
	$\displaystyle \bm{z}_{t}^* = \argmin_{\bm{z}\in\mathcal{Z}} {\color{salmon}\max_{\bm{p}\in\mathcal{P}(\delta)}}\ {\color{dodgerblue}f_{\bm{\epsilon}}(\bm{z}, {\color{salmon}\bm{p}})} - {\color{dodgerblue}\gamma}\cdot\hat{\bm{y}}_{t}^\top\bm{z}$
	};
	\node[align=center, anchor=south] (lab) at (C.north) {Decision layer};
	\draw [->,thick] (B.east) -- (C.west);
	\node[draw, rectangle, minimum width=8em, minimum height=4em, right= of C] (D) [below right=-2em and 0.75em of C.east] {$l\big(\bm{z}_{t}^*, \{\bm{y}_j\}_{j=t}^{t+v}\big)$};
	\node[align=center, anchor=south] (lab) at (D.north) {Task loss};
	\draw [->,thick] (C.east) -- (D);
	\node[draw, rectangle, text width= 4em, align=center, minimum width=5.5em, minimum height=4em, below= of A] (A2) {Update \\ $\bm{\theta}, {\color{dodgerblue}\gamma}, {\color{salmon}\delta}$
	};
	\draw[decorate, line width=0.1em, text width=2.5em, midway, align=center, decoration={brace,mirror,raise=0.5ex}] (A2.north west) -- (A2.south west)  node[midway,left=1ex]{Bwd pass} ;
	\node[draw, rectangle, minimum width=10.5em, minimum height=4em, below= of B] (B2) [below right=-2em and 0.75em of A2.east] {
	$
	\begin{aligned}
	\partial\hat{\bm{y}}_{t} &/ \partial \bm{\theta},\\ 
	{\color{dodgerblue}\partial\bm{\epsilon}}&{\color{dodgerblue}/\partial \bm{\theta}}	
	\end{aligned}
	$
	};
	\draw [<-,thick] (A2) -- (B2);
	\node[draw, rectangle, minimum width=18em, minimum height=4em, below= of C] (C2) [below right=-2em and 0.75em of B2.east] {
	$
	\begin{aligned}
	\partial\bm{z}_{t}^*/\partial \hat{\bm{y}}_{t},&\quad {\color{salmon}\partial\bm{z}_{t}^*/\partial \delta},\\
	{\color{dodgerblue}\partial\bm{z}_{t}^*/\partial \bm{\epsilon}},&\quad {\color{dodgerblue}\partial\bm{z}_{t}^*/\partial \gamma}	
	\end{aligned}
	$
	};
	\draw [<-,thick] (B2.east) -- (C2.west);
	\node[draw, rectangle, minimum width=8em, minimum height=4em, below= of D] (D2) [below right=-2em and 0.75em of C2.east] {
	$\partial l/\partial \bm{z}_{t}^*$
	};
	\draw [<-,thick] (C2.east) -- (D2);
\end{tikzpicture}
\caption{Computational graph of the end-to-end portfolio construction
  system. Each sample in the training set is composed of a batch of
  inputs--outputs $\{\bm{x}_j\}_{j=t-T}^{t}$, $\{\bm{y}_j\}_{j=t-T}^{t+v}$
  and yields a single optimal portfolio \(\bm{z}_{t}^*\). Here,
  \(f_{\bm{\epsilon}}\) is a measure of model prediction error and \(l\)
  is the task loss function. Note that the measure of task-based
  performance \(l\) from period \(t\) to period \(t+v\) is
  allowed to differ from the optimization objective in the decision
  layer. Our contributions are highlighted as follows. First, the
  components pertaining to the sample-based integration of the prediction
  error into the decision layer are shown in
  {\color{dodgerblue}blue}. Second, the robust design is shown in
  {\color{salmon}red}. Details of the notation are discussed later in
  Section \ref{sec:e2e}.} 
\label{fig:e2e}
\end{figure*}
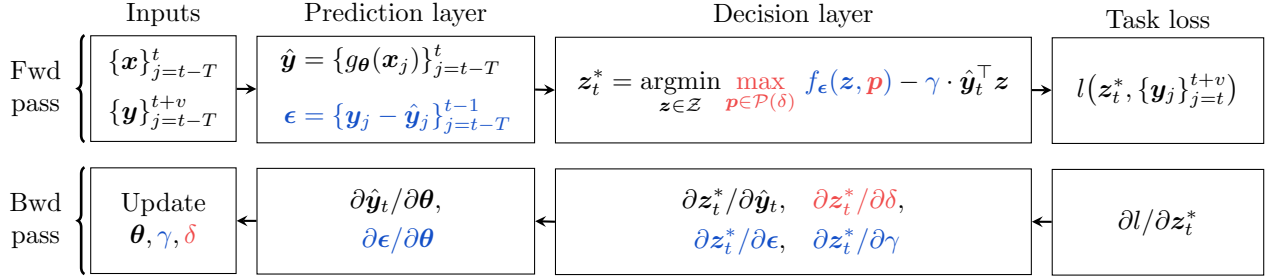

\subsection{Contributions}

Our main contribution is to show how to accommodate model robustness
within an end-to-end system in a tractable and intuitive fashion. The 
specific contributions of this paper are as follows (see also 
Figure~\ref{fig:e2e}):   

\begin{enumerate}[itemsep=-0.1em, topsep=0pt, leftmargin=*]
\item We propose an end-to-end portfolio construction system where 
    the decision layer is a \emph{distributionally} robust (DR) 
    optimization problem~(see the `Decision layer' box in 
    Figure~\ref{fig:e2e}). We show how to integrate the DR layer with 
    any prediction layer. 
  
\item The DR optimization problem requires both the point prediction 
    as well as the prediction errors as inputs to quantify and control 
    the model risk. Therefore, unlike standard end-to-end systems, we 
    provide both the point prediction, as well as a set of 
    past prediction errors as inputs to the decision layer~(see 
    the `Prediction layer' box in Figure~\ref{fig:e2e}). 
  
\item The DR layer is formulated as a minimax optimization problem where
    the objective function is a combination of the mean loss and the
    worst-case risk term that combines the impact of variability and model
    error. The worst-case risk is taken over a set of probability measures
    within a ``distance'' \(\delta\) of the empirical measure. We show that, 
    by using convex duality, the minimax problem can be recast as a minimization
    problem. In this form, the gradient of the overall task loss can be 
    back-propagated through the DR decision layer to the prediction layer.
    Moreover, this means the end-to-end system remains computationally tractable
    and can be efficiently solved using existing software. This result
    is of interest to embedding any minimax (or maximin) decision layer into an
    end-to-end system.
    
\item We show that the parameters that control the risk appetite~(\(\gamma\)) 
	and model robustness~(\(\delta\)) can be learned directly from data in our
    proposed DR end-to-end system. Learning \(\gamma\) is relatively straightforward. 
    However, showing that \(\delta\) can be learned in the end-to-end system is 
    non-trivial, and requires the use of convex duality. This step is very important 
    because setting appropriate values for these  parameters is, in practice, 
    difficult and computationally expensive~\citep{bertsimas2018data}. 
  
\item We have implemented our proposed DR end-to-end system for portfolio
    construction in Python. Source code is available at
    \href{https://github.com/Iyengar-Lab/E2E-DRO}{https://github.com/Iyengar-Lab}.  
\end{enumerate}

\subsection{Background and related work} 

Recently, there has been increasing
interest in end-to-end learning systems \citep{bengio1997using,
lecun2005off, thomas2006cognitive, elmachtoub2021smart}. The goal of an
end-to-end system is to integrate a prediction layer with the downstream
decision layer so that one estimates parameter values that minimize the
decision error (known as the `task loss'), instead of simply
minimizing the prediction error. The main technical challenge in such
systems is to back-propagate the gradient of the task loss through the
decision layer onto the prediction layer~\citep{amos2017optnet}. These
systems come in two varieties: model-free and model-based. Model-free
methods follow a `black-box' approach, and have found some success in
portfolio construction \citep{uysal2021end, zhang2020deepA}. However,
model-free methods are often data-inefficient during training. On the
other hand, model-based methods rely on some predefined structure of their
environment before model training can take place. Model-based methods have
the advantage of retaining some level of interpretability while also being
more data-efficient during training~\citep{amos2018differentiable}. 

Model-based end-to-end systems have found successful applications in asset
management and portfolio construction. In addition to the model-free system,
\citet{uysal2021end} also propose a model-based risk budgeting
portfolio construction model. \citet{butler2021integrating} propose an
end-to-end mean--variance optimization model. \citet{zhang2021universal}
integrate the convex optimization layer with a deep learning prediction
layer for portfolio construction. 

As a stand-alone tool, portfolio optimization has been criticized for its
sensitivity to model and parameter errors that leads to poor out-of-sample
performance~\citep{merton1980estimating, best1991sensitivity, chopra1993,
broadie1993computing}. Robust optimization methods that explicitly model
perturbation in the parameters of an optimization problem, and choose
decisions assuming the worst-case behavior of the parameters, have been
successfully employed to improve the performance of portfolio
selection~\citep[e.g., see][]{goldfarb2003robust, tutuncu2004robust,
fabozzi2007robust, costa2020robust}. The robust optimization approach
was subsequently extended to \emph{distributionally} robust optimization (DRO)
\citep{scarf1958min, delage2010distributionally, ben2013robust}, where the
parameters of an optimization problem are distributed according to a 
probability measure that belongs to a given ambiguity set. The DRO problem can 
be interpreted as a game between the decision-maker who chooses an action to
minimize cost, and an adversary~(i.e., nature) who chooses a
parameter distribution that maximizes the cost~\citep{neumann1928theorie}. 
The DRO approach has been implemented for portfolio optimization 
\citep[e.g., see][]{calafiore2007ambiguous, delage2010distributionally, 
costa2021datadriven}.  

We propose a model-based end-to-end portfolio selection framework where
the decision layer is 
a DR portfolio selection problem. In
keeping with existing end-to-end systems, we allow the prediction layer to
be any differentiable supervised learning model, ranging from simple
linear models to deep neural networks. However, unlike existing end-to-end
frameworks, the decision layer is a minimax problem. We show how to use
the errors from the prediction layer to construct this minimax problem. It
is well known that a differentiable optimization problem can be embedded
into the architecture of a neural network, allowing for gradient-based
information to be communicated between prediction and decision layers
\citep{amos2017optnet, donti2017task, agrawal2019differentiable, 
amos2019differentiable}. Using convex duality, we extend this result to minimax
problems, and show how to communicate the gradient information from the DR
decision layer back to the prediction layer.  

Finally, we note that the task loss function that guides the training of 
an end-to-end system does not need to be the same as the objective function 
in the decision layer. The discrepancy between these two functions may stem 
from different reasons. The decision layer may be designed as a 
computationally tractable surrogate for the task loss function. Alternatively, 
we may choose a task loss function that directs the system's training towards 
some desirable out-of-sample reward that cannot be explicitly embedded into 
the decision layer. For example, \citet{uysal2021end} present a decision 
layer designed to diversify financial risk; however, their system's task loss 
function emphasizes financial return. In such cases, the end-to-end system is 
fundamentally similar to a reinforcement learning problem. 

\section{End-to-end portfolio construction}\label{sec:e2e}

In this section, we present our proposed DR end-to-end portfolio construction 
system and describe each individual component of the system. To allow for a 
natural progression, the structure of this section follows the `forward pass' 
of Figure \ref{fig:e2e}, starting with a discussion of the prediction layer, 
followed by the DR decision layer, and finally the task loss function. We 
conclude by presenting the complete DR end-to-end algorithm.   

\subsection{Prediction layer}\label{sec:pred}

We consider portfolio selection in discrete time. Denote the present time as 
\(t\) and let \(\bm{x}_t \in \mathbb{R}^m\) denote \(m\) financial factors 
(i.e., predictive features) observed at time \(t\). Using these factors, we 
want to predict the random return \(\tilde{\bm{y}}_t \in \mathbb{R}^n\) on
\(n\) assets over the period \([t,t+1]\).  Let \(\{\bm{x}_{j} \in \mathbb{R}^m: j =
t-T, \ldots, t-1\}\) denote the historical time series of
financial factors and let \(\{\bm{y}_{j} \in \mathbb{R}^n: j = 
t-T,\dots,t-1\}\) denote the historical time series of returns on the
\(n\) assets with \(T\) time steps.

Suppose we have access to the features \(\bm{x}_t\). A prediction model 
\(g_{\bm{\theta}}: \mathbb{R}^m \rightarrow \mathbb{R}^n\) that maps \(\bm{x}_t\) 
to the prediction \(\hat{\bm{y}}_{t}\) of the expected return 
\(\mathbb{E}[\tilde{\bm{y}}_t]\) is assumed be a differentiable 
function of the parameter \(\bm{\theta}\); otherwise the model may be as simple or 
as complex as required. An illustrative example is the linear model,  
\[
	\hat{\bm{y}}_{t} \triangleq g_{\bm{\theta}}(\bm{x}_t) = \bm{\theta}^\top \bm{x}_t,
\]
where \(\hat{\bm{y}}_{t}\in\mathbb{R}^n\) is a prediction and
\(\bm{\theta}\in\mathbb{R}^{m\times n}\) is the matrix of weights for this
specific model. Note that the dimensions of \(\bm{\theta}\) will change
depending on the structure of the prediction model.   

Let \(\bm{\tilde{\epsilon}}_t \triangleq\ \bm{\tilde{y}}_t - \hat{\bm{y}}_t = 
\bm{\tilde{y}}_t - g_{\bm{\theta}}(\bm{x}_t) \in\mathbb{R}^n\) denote the prediction
error. Note that the prediction error is a combination of stochastic noise 
(i.e., variance) and model risk. We assume that the set of past prediction 
errors \(\{\bm{\epsilon}_j = \bm{y}_j - g_{\theta}(\bm{x}_j): j = t-T, \ldots, t-1\}\) 
are \(T\) IID samples of \(\tilde{\bm{\epsilon}}_{t}\). This sample set 
of prediction errors will be used to introduce distributional robustness in 
the decision layer.

Traditionally, prediction models are trained by minimizing a
prediction loss function to improve predictive accuracy. However, an
end-to-end system is concerned with minimizing the task loss rather than
the prediction loss. In our case, the task loss corresponds to some
measure of out-of-sample portfolio performance, which we will discuss in
Section~\ref{sec:taskloss}. For now, we will focus on how to use the set of
prediction errors to introduce distributional robustness into the decision
layer.  

\subsection{Decision layer}\label{sec:decision}

A feed-forward neural network is trained by iterating over `forward' and
`backward' passes. During the `forward pass', the network is evaluated
using the current set of weights. This is followed by the `backward pass',
where the gradient of the loss function is computed and propagated
backwards through the layers of the neural network in order to update the
weights.  

In existing end-to-end systems, during the forward pass, the decision
layer is treated as a standard convex  minimization
problem. On the other hand, the backward pass requires that we
differentiate through the `\(\argmin\)' operator \citep{donti2017task}. In
general, the solutions of optimization problems cannot be written as
explicit functions of the input parameters, i.e., they generally do not 
admit a closed-form solution that we can differentiate. However, it is 
possible to differentiate through a convex optimization problem by 
implicitly differentiating its optimality conditions, provided that some 
regularity conditions are satisfied \citep{amos2017optnet, 
agrawal2019differentiable}.  

First, we adapt the existing end-to-end system to solve a portfolio selection 
problem involving risk measures. This extension requires us to
work with the set of prediction errors \(\bm{\epsilon} = \{\bm{\epsilon}_j : 
j = t-T, \ldots, t-1\}\) in addition to the prediction \(\hat{\bm{y}}_{t}\) 
corresponding to the factor vector \(\bm{x}_t\). Next, we show how to extend the 
methodology to distributionally robust portfolio selection.

A portfolio is a vector of asset weights \(\bm{z}\in\mathcal{Z}\). In
order to keep the exposition simple, the set of admissible portfolios is 
\[
	\mathcal{Z} \triangleq \big\{ \bm{z}\in\mathbb{R}^n : \bm{z}\geq\bm{0},\ \bm{1}^\top \bm{z} = 1\big\}.
\]
The equality constraint in \(\mathcal{Z}\) is the budget constraint and
it ensures that the entirety of our available budget is invested, while
the non-negativity constraint disallows the short selling of financial
assets. Our methodology extends to sets defined by limits on
industry/sector exposures and leverage constraints.   

Suppose at time \(t\) we have access to the factors \(\bm{x}_{t}\) but
not to the realized asset returns \(\bm{y}_{t}\) over the period \([t,t+1]\). 
We approximate the expected return by the output \(\hat{\bm{y}}_{t} = 
g_{\bm{\theta}}(\bm{x}_{t})\) of the prediction layer. Next, we characterize 
the variability in the portfolio return both due to the stochastic noise 
(i.e. variance) and the model risk. Given that we allow the 
prediction layer to have any general form, we avoid attempting to measure 
the parametric uncertainty associated with the prediction layer weights. Instead, 
we take a data-driven approach and estimate the combined effect of variance
and model risk directly from a sample set of past prediction errors. We quantify 
the `risk' associated with the portfolio \(\bm{z}\) by a \emph{deviation risk}
measure~\citep{rockafellar2006generalized} defined below. 

\begin{proposition}[Deviation risk measure]
\label{prop:dev_risk}
Let \(\bm{\epsilon} = \{\bm{\epsilon}_j\in\mathbb{R}^n: j = 1, \ldots, T\}\) 
denote the finite set of prediction error outcomes and let 
\(\bm{z}\in\mathcal{Z} \subseteq \mathbb{R}^n\) denote a fixed portfolio. 
Suppose \(R:\mathbb{R}\rightarrow \mathbb{R}_+\cup +\infty\) is a
closed convex function where \(R(0) = 0\) and \(R(X) = R(-X)\).  Let
\(\bm{p}\) denote any probability mass function (PMF) in the probability
simplex 
\[
  \mathcal{Q}\triangleq \big\{\bm{p} \in\mathbb{R}^T : \bm{p}\geq \bm{0},\
  \bm{1}^\top \bm{p} = 1\big\}. 
\]
Then, the deviation risk measure \(f_{\bm{\epsilon}}(\bm{z}, \bm{p})\)
associated with the set of outcomes 
\(\bm{\epsilon}\), the 
portfolio
\(\bm{z}\) and PMF \(\bm{p}\) is given by  
\begin{equation}
\label{eq:dev_risk} 
f_{\bm{\epsilon}}(\bm{z}, \bm{p}) \triangleq \min_{c} \sum_{j=1}^T
p_j\cdot R(\bm{\epsilon}_j^\top \bm{z} - c). 
\end{equation}
The deviation risk measure \(f_{\bm{\epsilon}}(\bm{z}, \bm{p})\) has the following properties.
\begin{enumerate}[itemsep=-0.15em, topsep=0pt, leftmargin=*]
\item \(f_{\bm{\epsilon}}(\cdot, \bm{p}): \mathcal{Z} \rightarrow \mathbb{R}_+\) is convex for any fixed \(\bm{p}\in\mathcal{Q}\).
  
\item \(f_{\bm{\epsilon}}(\bm{z}, \bm{p})\geq 0\) for all \(\bm{z}\in\mathcal{Z}\) and \(\bm{p}\in\mathcal{Q}\).
   
\item \(f_{\bm{\epsilon}}(\bm{z}, \bm{p})\) is shift-invariant with respect to \(\bm{\epsilon}\).
  
\item \(f_{\bm{\epsilon}}(\bm{z}, \bm{p})\) is symmetric with respect to \(\bm{\epsilon}\).
\end{enumerate}
\end{proposition}
\begin{proof}
  See Appendix \ref{app:dev_risk}.
\end{proof}

Since our `deviation risk measure' pertains to the prediction error rather 
than financial risk, we use a broader definition of the risk measure as 
compared to traditional financial risk measures
\citep[e.g., see][]{rockafellar2006generalized}.

The centering parameter \(c\) plays an important role. It is crucial for
the shift invariant property (3), and this property in turn implies that
the deviation risk associated with the outcomes \(\bm{\epsilon}\) is the
same as that associated with the mean adjusted outcomes \(\{\bm{\epsilon}_j
- \sum_{k=1}^T p_k\cdot \bm{\epsilon}_k: j = 1, \ldots, T\}\). Thus, the
deviation risk measure is really a function of the deviations around the
mean. When \(R(X) = X^2\), the associated deviation risk measure
\[
    f_{\bm{\epsilon}}(\bm{z}, \bm{p}) = \sum_{j=1}^T p_j\cdot
    \bigg(\bm{\epsilon}_j^\top \bm{z} - \sum_{k=1}^T p_{k}\cdot
    \bm{\epsilon}_{k}^\top \bm{z}\bigg)^2
\]
is the variance, where the optimal 
\(c^* = \sum_{k=1}^T p_{k}\cdot\bm{\epsilon}_{k}^\top \bm{z}\) (see 
Appendix~\ref{app:variance} for details). Later, we introduce the worst-case 
deviation risk by taking the maximum over \(\bm{p} \in \mathcal{P}({\delta})\). 
In that setting, the centering parameter \(c\) ensures that the adversary 
cannot increase risk by putting all the weight on the 
worst~\(\bm{\epsilon}_j\).

\subsubsection{Nominal layer}

We start with the assumption that every outcome in
the set \(\bm{\epsilon} = \{\bm{\epsilon}_j : j = t-T, \ldots, t-1\}\) 
has equal probability. We refer to this as the nominal decision problem.
Let \(\bm{q}\in\mathcal{Q}\) denote a uniform probability distribution
(i.e., \(q_j = 1/T\ \forall j\)). Then, the nominal decision layer
computes 
\begin{equation}
\label{eq:nom_opt}
	\bm{z}_{t}^* = \argmin_{\bm{z}\in\mathcal{Z}}\
        f_{\bm{\epsilon}}(\bm{z}, \bm{q}) - \gamma\cdot
        \hat{\bm{y}}_{t}^\top \bm{z} 
\end{equation}
where \(\gamma \in\mathbb{R}_+\) is the risk appetite parameter and
\(f_{\bm{\epsilon}}\) is a deviation risk measure as defined by
Proposition \ref{prop:dev_risk}.   

A challenge often faced by practitioners is determining an appropriate
value for \(\gamma\). The parameter is typically calibrated by
trial-and-error using in-sample performance, potentially biasing the
performance quite heavily. Instead, we treat \(\gamma\) as a learnable 
parameter of the end-to-end system. As shown by \cite{amos2017optnet} 
and \citet{agrawal2019differentiable}, we can find the partial derivative 
of the task loss function with respect to \(\gamma\) and subsequently 
use it to update \(\gamma\) through gradient descent like any other 
parameter in the network. This is advantageous because we are able to learn
a task-optimal value of \(\gamma\) and relieve the user from
having to compute an appropriate value of \(\gamma\).

\subsubsection{DR layer}\label{sec:dr_layer}

The nominal problem puts equal weight on every sample in 
the set \(\bm{\epsilon}\), i.e. the PMF \(\bm{p}\) defining the risk measure 
is the uniform distribution \(\bm{q}\). We introduce distributional
robustness by allowing \(\bm{p}\) to deviate from \(\bm{q}\) to within a
certain ``distance'' \citep{calafiore2007ambiguous, ben2013robust, kannan2020residuals,
costa2021datadriven}. Let \(I_\phi \triangleq \sum_{j=1}^T q_j\cdot
\phi(p_j/q_j):  \mathcal{Q} \times \mathcal{Q}
\rightarrow \mathbb{R}_+\) denote a statistical distance  function
based on a \(\phi\)-divergence (e.g., Kullback--Leibler, Hellinger, chi-squared).
Then, the ambiguity set \(\mathcal{P}(\delta)\) for
the distribution \(\bm{p}\) is given by 
\[
  \mathcal{P}(\delta) \triangleq \big\{\bm{p} \in\mathbb{R}^T :
  \bm{p}\geq \bm{0},\ \bm{1}^\top \bm{p} = 1,\ I_\phi(\bm{p},\bm{q})
  \leq \delta\big\}.
\]
The size parameter \(\delta\) defines the maximum permissible distance
between \(\bm{p}\) and the nominal distribution \(\bm{q}\). The DR
decision layer chooses the portfolio \(\bm{z}^{*}_{t}\) assuming the worst
case behavior of \(\bm{p} \in \mathcal{P}({\delta})\), i.e. it solves the
following minimax problem: 
\begin{equation}
\label{eq:dr_minimax}
	\bm{z}_{t}^* = \argmin_{\bm{z}\in\mathcal{Z}}\ \max_{\bm{p}\in\mathcal{P(\delta)}} f_{\bm{\epsilon}}(\bm{z}, \bm{p}) - \gamma\cdot \hat{\bm{y}}_{t}^\top \bm{z}.
\end{equation}
In general, solving minimax problems can be computationally expensive and
can also lead to local optimal solutions rather than global solutions.
It is straightforward to check that the objective in \eqref{eq:dr_minimax}
is convex in \(\bm{z}\); however, the concavity property of the deviation 
risk measure \(f_{\bm{\epsilon}}(\bm{z}, \bm{p}) \) as a function of \(\bm{p}\) 
is not immediately obvious. Nevertheless, we show in Appendix~\ref{app:dual} 
that we can still use convex duality to reformulate the minimax 
problem~\eqref{eq:dr_minimax} into the following
minimization problem, 
\begin{equation}
\label{eq:dr_opt}
 \bm{z}_{t}^* = \argmin_{\bm{z}\in\mathcal{Z},\ \lambda\geq 0,\
   \xi,\ c} f_{\bm{\epsilon}}^{\delta}(\bm{z}, c, \lambda, \xi) -
    \gamma \cdot \hat{\bm{y}}_{t}^\top \bm{z},
\end{equation}
where
\begin{align}
	&f_{\bm{\epsilon}}^{\delta}(\bm{z}, c, \lambda, \xi)\nonumber \\
  	&\quad\qquad \triangleq \xi +
  	\delta\cdot \lambda + \frac{\lambda}{T} \sum_{j=1}^T \phi^*\bigg(
  	\frac{R(\bm{\epsilon}_j^\top \bm{z} - c) -
  	\xi}{\lambda}\bigg),\label{eq:dr_obj}
\end{align}
\(\phi^*(\cdot)\) is the convex conjugate of the \(\phi\)-divergence
defining the ambiguity set \(\mathcal{P}({\delta})\),\footnote{For a
description of \(\phi\)-divergence functions and their convex
conjugates, please refer to Tables 2 and 4 in \citet{ben2013robust}.}
and \(\lambda\geq 0, \xi\in\mathbb{R}\) are auxiliary variables arising
from constructing the Lagrangian dual. Note that we have abused our 
notation of the `\(\argmin\)' operator in \eqref{eq:dr_opt} since 
\(\bm{z}_{t}^*\) is the only pertinent output.

Tractable reformulations of the DR layer exist for many choices of 
\(\phi\)-divergence. In Appendix \ref{app:reform}, we show that 
the DR layer can be formulated as 
a second-order cone problem when \(\phi\) is the Hellinger distance\footnote{The DR layer reduces to a 
second-order cone program if the function \(R(X)\) in 
Proposition~\ref{prop:dev_risk} is quadratic or piecewise linear. 
Otherwise, the complexity of the problem is dictated by the choice of 
\(R(X)\).} and as a linear optimization problem when \(\phi\) is the 
Variational distance\footnote{The DR layer 
reduces to a linear program if the function \(R(X)\) in 
Proposition~\ref{prop:dev_risk} is piecewise linear. Otherwise, the 
complexity of the problem is dictated by the choice of \(R(X)\).}.  
\citet{ben2013robust} provide tractable reformulations for other
choices of \(\phi\)-divergence.

Recasting the minimax problem into a convex minimization problem 
allows us to differentiate through the DR layer during training of the 
end-to-end system \citep{amos2017optnet, amos2019differentiable}. Another 
benefit of dualizing the inner maximization problem is that the ambiguity 
sizing parameter \(\delta\) becomes an explicit part of the DR layer's 
objective function, and can, therefore,  be treated as a learnable parameter 
of the end-to-end system. Determining the size of an ambiguity set a priori
is often a subjective exercise, with many users resorting to a probabilistic 
confidence level. By treating \(\delta\) as a learnable parameter, we relieve 
the user from the responsibility of having to assign a value of \(\delta\) a 
priori.

\subsection{Task loss}\label{sec:taskloss}

In standard supervised learning models, the loss function is a 
measure of predictive accuracy. For example, a popular prediction loss 
function is the mean squared error (MSE). For a prediction at time \(t\), the 
loss is  
\begin{equation}
\label{eq:loss_mse}
	l_{\text{mse}}(\hat{\bm{y}}_t, \bm{y}_t) \triangleq 
	\frac{1}{n}\|\bm{y}_t - \hat{\bm{y}}_t \|_2^2.
\end{equation}
In an end-to-end system, predictive loss measures the performance of only 
the the prediction layer. However, using the predictive loss to measure 
the performance of the entire system fails to consider our main objective: 
the out-of-sample performance of the \emph{decision}.

Therefore, in contrast to standard supervised learning models, end-to-end systems 
measure their performance using a `task loss' function, which is chosen in 
order to train the system based on the out-of-sample performance of the 
optimal decision. For example, the task loss in \citet{butler2021integrating} 
has the same form as the objective function of the decision layer, except 
the predictions are replaced with the corresponding realizations in order 
to calculate the out-of-sample performance. We allow for the possibility 
that the task-loss is different from the objective function of the decision 
layer. This allows us to approximate a task loss function that is hard to 
optimize with a more tractable surrogate. In such cases, the end-to-end 
system is fundamentally similar to a reinforcement learning problem.  

We define the task loss as the financial performance of our optimal
portfolio \(\bm{z}_{t}^*\) measured over some out-of-sample period of 
length \(v+1\) with the realized asset returns
\(\{\bm{y}_j : j= t, \ldots,t+v\}\). When \(v = 0\), we set the task 
loss to the realized return \(\bm{y}_t^\top \bm{z}_t^*\). When 
\(v > 0\), we can use other measures of financial performance, e.g. the
portfolio volatility or the Sharpe ratio~\citep{sharpe1994sharpe} 
calculated over the next \(v\) time steps. In our numerical experiments, 
we set the task loss to be a weighted combination of the predictive accuracy
and the Sharpe ratio\footnote{Note that we have defined
the Sharpe ratio using the portfolio returns rather than the portfolio
\emph{excess} returns (i.e., the returns in excess of the risk-free
rate).}
\begin{equation}
\label{eq:loss_perf_sharpe}
	l_{\text{SR}}\big(\bm{z}_{t}^*, \{\bm{y}_j\}_{j=t}^{t+v}\big) \triangleq -\frac{\text{mean}\big(\{\bm{y}_{j}^\top \bm{z}_{t}^*\}_{j=t}^{t+v}\big)}{\text{std}\big(\{\bm{y}_{j}^\top \bm{z}_{t}^*\}_{j=t}^{t+v}\big)},
\end{equation}
where the operator \(\text{mean}(\cdot)\) calculates the mean of a set,
while \(\text{std}(\cdot)\) calculates the standard deviation.   
 
In general, the task loss is any differentiable function of the optimal
portfolio \(\bm{z}_{t}^*\). In turn, \(\bm{z}_{t}^*\) is an implicit
function of the parameters \(\bm{\theta}\), \(\gamma\) and
\(\delta\). Therefore, during the backward training pass, we can
differentiate the task loss with respect to \(\gamma\) and \(\delta\) in
the decision layer, and with respect to \(\bm{\theta}\) in the prediction
layer.   

\subsection{Training the DR end-to-end system}\label{sec:algo}

Our proposed DR end-to-end system for portfolio construction is detailed
in Algorithm \ref{algo:DRe2e}. Recall that \(T\) denotes the number of
prediction error samples from which to estimate the portfolio's prediction
error, while \(v\) is the length of the out-of-sample performance
window. Additionally, let us define \(T_0\) as the total number of
observations in the full training data set.

\begin{algorithm}[tb]
\caption{DR end-to-end system training}
\label{algo:DRe2e}
	\KwIn{\(\{\bm{x}_j\}_{j=1}^{T_0-v}\), \(\{\bm{y}_{j}\}_{j=1}^{T_0}\)}
	\KwInit{system parameters \(\bm{\theta}\), \(\gamma\), \(\delta\); learning rate \(\eta\); number of epochs \(K\)}
	\For{\(k = 1,\dots,K\)}{
		\(L\leftarrow 0\)\;
		\For{\(t = T+1,\dots,T_0-v\)}{
			\makebox[2.5em][r]{\(\displaystyle \hat{\bm{y}}\)} \(\displaystyle\leftarrow \{\bm{g}_{\bm{\theta}}(\bm{x}_j) : j = t-T, \ldots, t \}\)\;
			\makebox[2.5em][r]{\(\displaystyle \bm{\epsilon}\)} \(\displaystyle\leftarrow \{\bm{y}_j - \hat{\bm{y}}_j : j = t-T, \ldots, t-1\}\)\;
			\makebox[2.5em][r]{\(\displaystyle \bm{z}_{t}^*\)} \(\displaystyle\leftarrow \argmin_{\bm{z}\in\mathcal{Z}, \lambda\geq 0, \xi, c} f_{\bm{\epsilon}}^{\delta}(\bm{z}, c, \lambda, \xi) - \gamma \cdot \hat{\bm{y}}_{t}^\top \bm{z}\)\;
			\makebox[2.5em][r]{\(\displaystyle L\)} \(\displaystyle\leftarrow L + \frac{1}{T_0-T-v}l_{\text{task}}\big(\bm{z}_{t}^*, \{\bm{y}_j\}_{j=t}^{t+v}\big)\)\;
		}
		Update \(\bm{\theta}, \gamma, \delta\) with \(\nabla_{\bm{\theta}} L\), \(\nabla_{\gamma} L\), \(\nabla_{\delta} L\)\;
    }
\KwOut{\(\bm{\theta}\), \(\gamma\), \(\delta\)}
\end{algorithm}

In certain settings, a user may be unable or unwilling to integrate 
the prediction layer with the rest of the system. 
Alternatively, they may be using a prediction layer that cannot be 
trained via gradient descent, e.g. a tree-based predictor. In such 
cases, we assume the prediction layer is fixed, which in turn means 
that the prediction errors are constant during training.  
 
Nevertheless, we can still pass a sample set of \(T\) prediction errors 
as an input to the DR layer during training. Since the DR layer is a 
differentiable convex optimization problem, we are still able to learn 
values of \(\gamma\) and \(\delta\) that minimize the task loss. In 
practice, setting the risk appetite parameter \(\gamma\) and model 
robustness parameter \(\delta\) is difficult, and requires significant 
effort. For example, \citet{bertsimas2018data} propose using 
cross-validation techniques to set the level of robustness. Instead, 
our end-to-end system learns these parameters directly from data as 
part of the overall training in a much more efficient manner.

\section{Numerical experiment}\label{sec:num_ex}
We present the results of five numerical experiments. Each experiment 
evaluates different characteristics of our proposed DR end-to-end system. 
The first four experiments are conducted using historical data from
the U.S. stock market. The fifth experiment uses synthetic data generated
from a controlled stochastic process.  

The first experiment provides a holistic measure of financial performance. 
The second, third and fourth experiments isolate the out-of-sample effect 
of allowing the end-to-end system to learn the parameters \(\gamma\), 
\(\delta\) and \(\bm{\theta}\), respectively. Finally, the fifth experiment 
evaluates the effect of robustness when working with complex prediction 
layers.

The numerical experiments were conducted using a code written in
Python (version 3.8.5), with PyTorch (version 1.10.0) 
\citep{paszke2019pytorch} and Cvxpylayers (version 0.1.4) 
\citep{agrawal2019differentiable} used to compute the
end-to-end systems. The neural network is trained using the `Adam'
optimizer \citep{kingma2014adam} and the `ECOS' convex optimization 
solver \citep{domahidi2013ecos}.

\subsection{Competing investment systems}\label{sec:ex_models}

The numerical experiments involve seven different investment systems. 
The individual experiments compare these systems against each other. 
Although many of the systems are designed to learn the parameters \(\bm{\theta}\), 
\(\gamma\) and \(\delta\), some experiments purposely keep these parameters 
constant in order to isolate the effect of learning the remaining parameters. 
The seven investment systems are described below.

\begin{enumerate}[itemsep=-0.15em, topsep=0pt, leftmargin=*]
  
\item \emph{Equal weight (EW)}: Simple portfolio where all assets
  have equal weight -- no prediction or optimization is required
  and no parameters need to be learned. Equal weight portfolios 
  promote diversity and have been empirically shown to have a good 
  out-of-sample Sharpe ratio~\citep{demiguel2009optimal}.  
  
\item \emph{Predict-then-optimize (PO)}: Two-stage system with a linear
  prediction layer. The decision layer is given by the nominal problem defined 
  in~\eqref{eq:nom_opt} with \(\gamma\) held constant. No parameters are 
  learned (i.e., once the parameters are initialized, they are held constant).  
  
\item \emph{Base}: End-to-end system that does \emph{not} incorporate the
  risk function \(f_{\bm{\epsilon}}(\bm{z},\bm{p})\) and chooses
  portfolios by solving the optimization problem
  \[
    \bm{z}_{t}^* = \argmax_{z\in\mathcal{Z}}\ \hat{\bm{y}}_{t}^\top \bm{z}. 
  \]
  The prediction layer is linear and the only learnable parameter is
  \(\bm{\theta}\). Note that the base system is equivalent to a system where 
  the variability of the outcome is not impacted by the decision \(\bm{z}_t\) 
  -- as was the case in~\citet{donti2017task}.
  
\item \emph{Nominal}: End-to-end system with a linear prediction layer and
  a decision layer corresponds to the nominal problem~\eqref{eq:nom_opt}. The
  learnable parameters are \(\bm{\theta}\) and \(\gamma\).
  
\item \emph{DR}: Proposed end-to-end system with a linear prediction layer 
  and a decision layer corresponds to the DR problem~\eqref{eq:dr_opt}. We 
  choose the Hellinger distance as the \(\phi\)-divergence to define the 
  ambiguity set \(\mathcal{P}(\delta)\). The learnable parameters are 
  \(\bm{\theta}\), \(\gamma\) and \(\delta\).  

\item \emph{NN-nominal}: End-to-end system with a non-linear prediction 
  layer. The prediction layer is composed of a neural network with either 
  two or three hidden layers. The decision layer corresponds to the nominal
  problem~\eqref{eq:nom_opt}. The learnable parameters are \(\bm{\theta}\) 
  and \(\gamma\).

\item \emph{NN-DR}: End-to-end system with a non-linear prediction layer. 
  The prediction layer is composed of a neural network with either two or 
  three hidden layers. The decision layer corresponds to the DR 
  problem~\eqref{eq:dr_opt}. The learnable parameters are \(\bm{\theta}\), 
  \(\gamma\) and \(\delta\).
  
\end{enumerate}

Additionally, since retaining some degree of predictive accuracy in the 
prediction layer is often desirable \citep{donti2017task}, we define the task 
loss as a linear combination of the Sharpe ratio loss 
in~\eqref{eq:loss_perf_sharpe} and the MSE loss in~\eqref{eq:loss_mse},
\begin{align}
    &l_{\text{task}}\big(\bm{z}_{t}^*, \hat{\bm{y}}_t, 
    \{\bm{y}_j\}_{j=t}^{t+v}\big)\nonumber\\
    &\qquad\quad = 0.5\cdot l_{\text{mse}}(\hat{\bm{y}}_t, \bm{y}_{t}) 
    + l_{\text{SR}}\big(\bm{z}_{t}^*, \{\bm{y}_j\}_{j=t}^{t+v}\big).\nonumber
\end{align}
The look-ahead evaluation period of the task loss from \(t\) to \(t+v\)
consists of one financial quarter (i.e., 13 weeks), which means
\(v=12\). The weight of \(0.5\) on the MSE loss was chosen to ensure
a reasonable trade-off between out-of-sample performance and prediction 
error. We expect similar performance with other weights.

\subsection{Experiments with historical data}\label{sec:exp_hist}

The historical data consisted of weekly asset and feature returns from 
07--Jan--2000 to 01--Oct--2021. The predictive feature data were sourced 
from Kenneth French's database~\citep{French2021Data} and consist of the 
weekly returns of eight features. The asset data were sourced from AlphaVantage
(www.alphavantage.co) and consist of the weekly returns of 20 U.S. stocks 
belonging to the S\&P~500 index. The selected features and assets are listed in 
Table~\ref{table:assets}. To avoid prediction biases, the input--output pairs 
are lagged by a period of one week, e.g. asset returns \(\bm{y}\)
for 14--Jan--2000 are predicted using the feature vector \(\bm{x}\) observed
on 07--Jan-2000.

\begin{table}[!ht]
\caption{List of features and assets}
\centering
\begin{tabular}{lllll}
\toprule
	\multicolumn{5}{c}{Features} \\
\midrule
	Market & \multicolumn{2}{l}{\quad Profitability} & \multicolumn{2}{l}{Investment}\\[0.5ex]
	Size & \multicolumn{2}{l}{\quad ST reversal} & \multicolumn{2}{l}{LT reversal}\\[0.5ex]
	Value & \multicolumn{2}{l}{\quad Momentum}\\[0.5ex]
\toprule
	\multicolumn{5}{c}{Assets} \\
\midrule
	AAPL & MSFT & AMZN & C   & JPM \\[0.5ex]
	BAC  & XOM  & HAL  & MCD & WMT\\[0.5ex]
	COST & CAT  & LMT  & JNJ & PFE\\[0.5ex]
	DIS  & VZ   & T    & ED  & NEM\\[0.5ex]
\bottomrule
\end{tabular}
\label{table:assets}
\end{table}

For each experiment, the setup of the participating investment systems is 
outlined in a corresponding table (see Table~\ref{table:exp1_models} as an
example). These tables indicate the initial values for the parameters 
\(\bm{\theta}\), \(\gamma\) and \(\delta\), and whether these parameters
were learned during training or they remained constant during the 
experiment. Some experiments were designed to isolate the effect of 
learning a specific parameter; in these experiments, other parameters were
kept constant even when the investment system could potentially learn 
them.

The training of the end-to-end learning systems was carried out
as follows. We used the portfolio error variance as the deviation risk measure 
\(f_{\bm{\epsilon}}\) for all investment systems (we show in 
Appendix~\ref{app:variance} that the variance can be cast as a deviation 
risk measure \(f_{\bm{\epsilon}}\)). For consistency, all linear prediction 
layers are initialized to the ordinary least squares (OLS) weights. Additionally,
the initial values of the  risk appetite parameter \(\gamma\) and 
robustness parameter \(\delta\) were sampled uniformly from appropriately
defined intervals (see Appendix~\ref{app:init} for the initialization 
methodology).

We used two years of weekly observations as a representative sample of 
prediction errors (\(T=104\)). The data is separated by a \(60:40\) 
ratio into training and testing sets, respectively. The training set 
ranges from 07--Jan--2000 to 18--Jan--2013 and serves to train and 
validate the end-to-end systems.  

Cross-validation was used to tune the learning rate \(\eta\) and the number 
of training epochs \(K\). Since we are working 
with time series data, the hyperparameters were selected through a time 
series split cross-validation process. To do this, we trained and validated 
the end-to-end systems over four separate folds. For each fold, the 
original training set was divided into a training subset and a validation 
subset. We began by using the first 20\% of the training set as the 
training subset, and the subsequent 20\% as the validation set. This is 
increased to a ratio of 40:20, 60:20 and, finally, 80:20. Moreover, we 
tested all possible combinations between three possible learning rates, 
\(\eta\in\{0.005, 0.0125, 0.02\}\), and number of epochs 
\(K\in\{30, 40, 50, 60, 80, 100\}\).  
 
Once all four folds were completed, the average validation loss was 
calculated and used to select the optimal hyperparameters that yield 
the lowest validation loss for each end-to-end system. Once the optimal
hyperparameters were selected, they were kept constant during the 
out-of-sample test. The average validation losses of the end-to-end 
systems from Experiments 1--4 are presented in Table~\ref{table:validation} 
in Appendix~\ref{app:add_exp}. The table also highlights the optimal 
hyperparameters selected for each system. 

The out-of-sample test ranges from 25--Jan--2013 to 
01--Oct--2021, with the first prediction taking place on 25--Jan--2013. 
Immediately before the test starts, the portfolios were retrained using 
the full training data set. Note that the hyperparameters remained constant 
with the values selected during the cross-validation stage. 

Following the best practice for time series models, we retrained the
investment systems approximately every two years using all past data 
available at the time of training. Therefore, the investment systems 
were trained a total of four times. Before retraining takes place, 
the prediction layer weights are reset to the OLS weights computed from
the corresponding training data set. In addition, the parameters \(\gamma\) 
and \(\delta\) are reset to their initial values before retraining takes place.

For each experiment, the out-of-sample results are presented in the form of 
a wealth evolution plot, as well as a table that summarizes the financial 
performance of the competing portfolios.

\subsubsection{Experiment 1: General evaluation}\label{sec:gen_eval}

The first experiment is a
complete financial ``backtest'' to evaluate 
the performance of the DR end-to-end learning system as an asset management 
tool. To do so, the DR system is compared against four other competing
systems presented in Table~\ref{table:exp1_models}.
In this set of experiments, the systems
were able to learn all available parameters within their purview.

\begin{table}[t]
\caption{Experiment 1 -- List of models}
\centering
\begin{tabular}{lrcr@{}rcr@{}rc}
\toprule
		& \multicolumn{2}{c}{\(\bm{\theta}\)} & & \multicolumn{2}{c}{\(\gamma\)} & & \multicolumn{2}{c}{\(\delta\)}\\[0.5ex] \cline{2-3} \cline{5-6} \cline{8-9}
\rule{0pt}{3ex}System 	& Val. 	& Lrn    && Val. & Lrn 	  && Val. & Lrn \\
\midrule
EW 		& - 	& -      && -     & -      && -     & -      \\[0.5ex]
PO 		& OLS 	& -      && 0.046 & -      && -     & -      \\[0.5ex]
Base	& OLS 	& \cmark && -     & -      && -     & -      \\[0.5ex]
Nominal & OLS 	& \cmark && 0.046 & \cmark && -     & -      \\[0.5ex]
DR		& OLS 	& \cmark && 0.046 & \cmark && 0.312 & \cmark \\[0.25ex]
\bottomrule
\multicolumn{9}{p{0.9\linewidth}}{\small\rule{0pt}{3ex}\textbf{Note(s)}: Val, Initial value; Lrn, Learnable.}\\
\end{tabular}
\label{table:exp1_models}
\end{table}

The out-of-sample financial performance of the five competing investment 
systems is compared as follows:
Figure~\ref{fig:exp1_wealth} shows the 
wealth evolution of the five corresponding portfolios, Figure~\ref{fig:exp1_sr} 
shows the cumulative Sharpe ratio, and 
Table~\ref{table:exp1_results} presents 
a summary of the results over the complete investment horizon. The experimental 
results lead to the following observations. 

\begin{itemize}[itemsep=-0.15em, topsep=0pt, leftmargin=*]
  
\item \emph{Benefit of using a sample-based model risk measure}: Unlike the base 
  model, the nominal and DR systems integrate a sample-based prediction error 
  into the decision layer. The results in Table~\ref{table:exp1_results} clearly 
  show the significance of incorporating prediction errors into the decision 
  layer -- both the nominal and DR systems have a higher return and lower 
  volatility than the base system.
  
\item \emph{Impact of end-to-end learning}: When compared against the 
  straightforward predict-then-optimize and equal weight systems, the nominal 
  and DR systems have higher Sharpe ratios on average over the entire
  investment horizon, highlighting the advantage of end-to-end systems of 
  being able to learn the prediction and decision parameters. We note that the nominal and 
  DR portfolios have a higher volatility due to their pronounced growth in 
  wealth, but nevertheless maintain a high cumulative Sharpe ratio as 
  shown in Figure~\ref{fig:exp1_sr}.
  
\item \emph{Distributional robustness}: Comparing the nominal and DR 
  systems, we clearly see the benefit of incorporating robustness into our 
  portfolio selection system. The DR system has a higher Sharpe ratio, which is 
  a reflection of our choice of a task loss function. 
  
\end{itemize}

\begin{figure}[ht]
\begin{center}
\centerline{\includegraphics[width=\columnwidth]{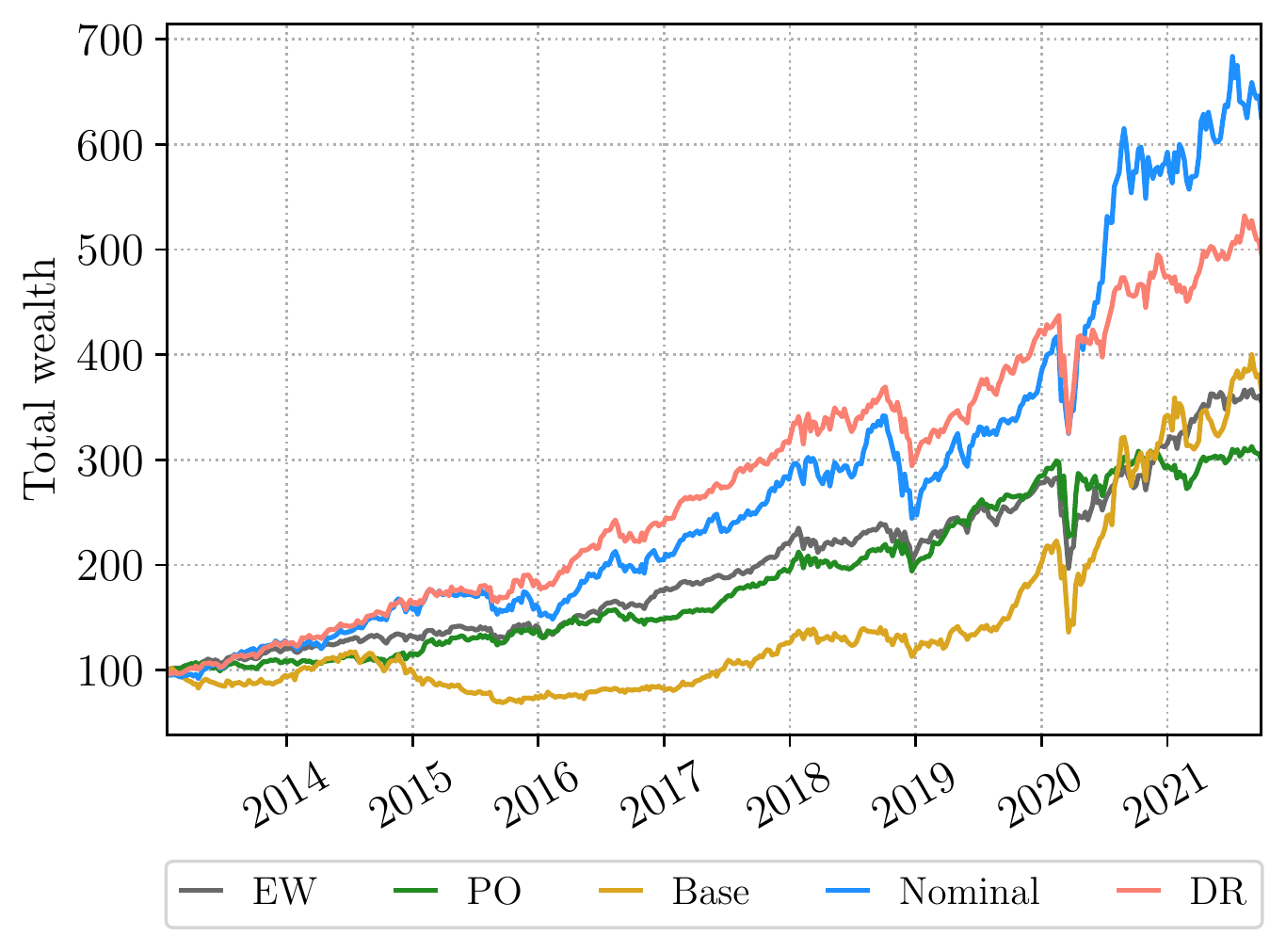}}
\caption{Experiment 1 -- Wealth evolution}
\label{fig:exp1_wealth}
\end{center}
\vskip -0.2in
\end{figure}

\begin{figure}[ht]
\begin{center}
\centerline{\includegraphics[width=\columnwidth]{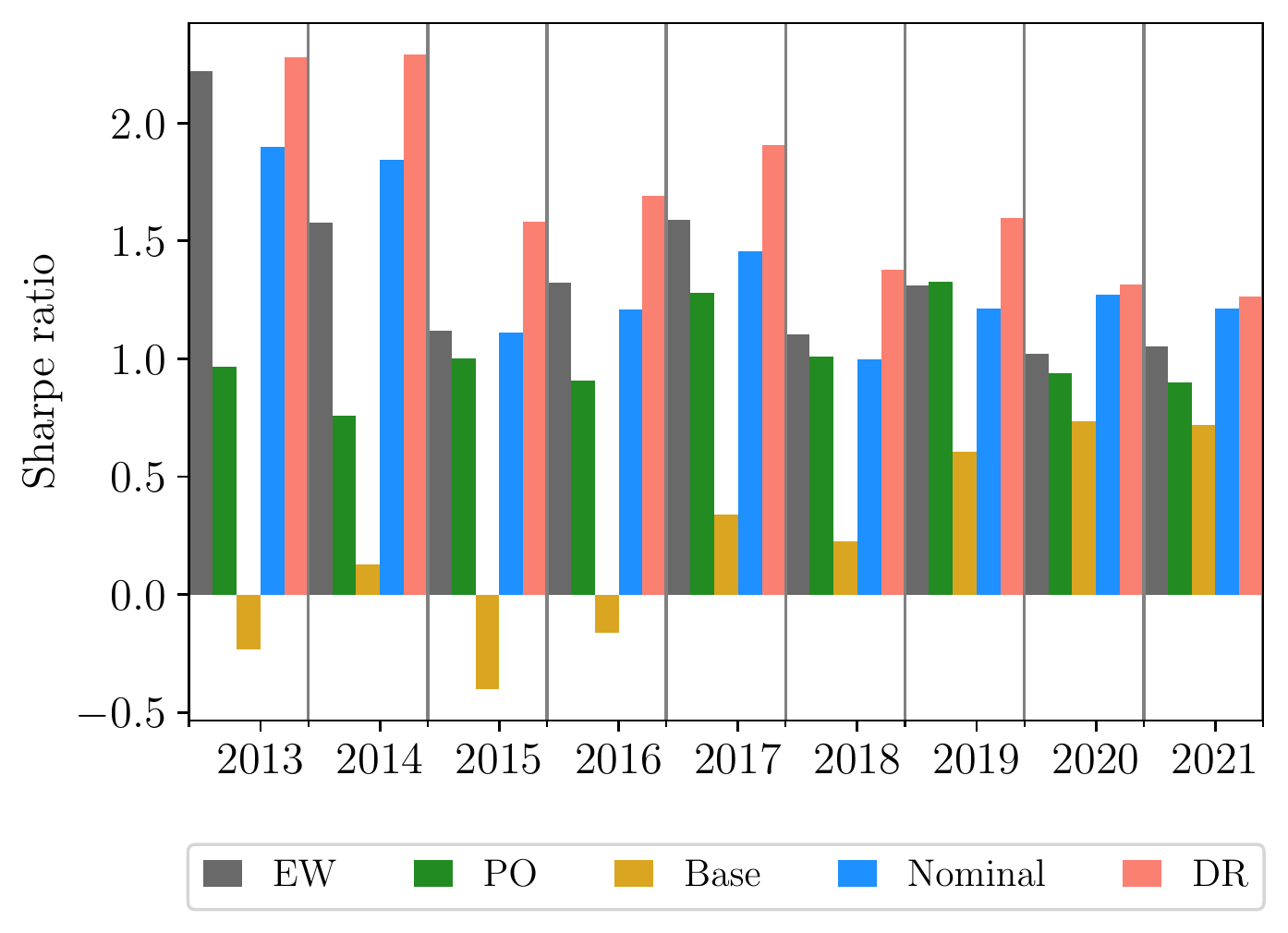}}
\caption{Experiment 1 -- Cumulative Sharpe ratio. The cumulative Sharpe ratio is computed at year-end. Values are annualized.}
\label{fig:exp1_sr}
\end{center}
\vskip -0.2in
\end{figure}

\begin{table}[t]
\caption{Experiment 1 -- Results}
\centering
\begin{tabular}{lrrrrr}
\toprule
				& EW & PO & Base & Nom. & DR\\
\midrule
Return (\%)     & 15.6 & 13.4 & 16.1 & 23.4 & 20.1 \\
Volatility (\%) & 14.9 & 15.3 & 25.0 & 18.8 & 15.5 \\
Sharpe ratio    & 1.05 & 0.88 & 0.64 & 1.24 & 1.30 \\
\bottomrule
\multicolumn{5}{p{0.6\linewidth}}{\small\rule{0pt}{3ex}\textbf{Note}: Values are annualized.}\\
\end{tabular}
\label{table:exp1_results}
\end{table}

\subsubsection{Experiment 2: Learning \texorpdfstring{\(\delta\)}{Lg}}\label{sec:learn_delta}

The second experiment focused on understanding the improvement in 
performance by learning \(\delta\); therefore, \(\bm{\theta}\) and 
\(\gamma\) were fixed. We focused on three investment systems: the 
predict-then-optimize~(PO) system (since \(\bm{\theta}\) and \(\gamma\) 
are fixed, this is the same as the nominal system), the DR system with 
constant \(\bm{\theta}\), \(\gamma\) and \(\delta\); and the  DR
system  with constant \(\bm{\theta}\) and \(\gamma\), but with a learnable
\(\delta\). The difference in performance between the PO system and the DR
system with all parameters fixed will highlight the benefit of adding some
(not optimized) robustness, and the difference in performance between the DR
systems with optimized \(\delta\) and fixed \(\delta\) is a measure of the
impact of size of the uncertainty set. Details of the three systems are 
presented in Table~\ref{table:exp2_models}. The out-of-sample
financial performance of the three investment systems is summarized
in Figure~\ref{fig:exp2_wealth} and Table~\ref{table:exp2_results}. 
The experimental results lead to the following observations.

\begin{itemize}[itemsep=-0.15em, topsep=0pt, leftmargin=*]

\item \emph{Impact of robustness without learning}: The results for the 
  PO system and the DR with no parameter learning show that, even when 
  learning is not allowed, robustness has a positive impact on out-of-sample 
  performance.   
  
\item \emph{Isolated learning of \(\delta\)}: The results in 
  Table~\ref{table:exp2_results} show that, although adding robustness 
  improves the Sharpe ratio, solely optimizing \(\delta\) can be detrimental 
  to the system's out-of-sample performance.
  
\end{itemize}

\begin{table}[t]
\caption{Experiment 2 -- List of models}
\centering
\begin{tabular}{lrcr@{}rcr@{}rc}
\toprule
		& \multicolumn{2}{c}{\(\bm{\theta}\)} & & \multicolumn{2}{c}{\(\gamma\)} & & \multicolumn{2}{c}{\(\delta\)}\\[0.5ex] \cline{2-3} \cline{5-6} \cline{8-9}
\rule{0pt}{3ex}System 	& Val. 	& Lrn    && Val. & Lrn 	  && Val. & Lrn \\
\midrule
PO & OLS & -      && 0.046 & -      && -     & -      \\[0.5ex]
DR & OLS & \xmark && 0.046 & \xmark && 0.312 & \xmark \\[0.5ex]
DR & OLS & \xmark && 0.046 & \xmark && 0.312 & \cmark \\[0.25ex]
\bottomrule
\multicolumn{9}{p{0.9\linewidth}}{\small\rule{0pt}{3ex}\textbf{Note}: Val, Initial value; Lrn, Learnable.}\\
\end{tabular}
\label{table:exp2_models}
\end{table}

\begin{figure}[ht]
\begin{center}
\centerline{\includegraphics[width=\columnwidth]{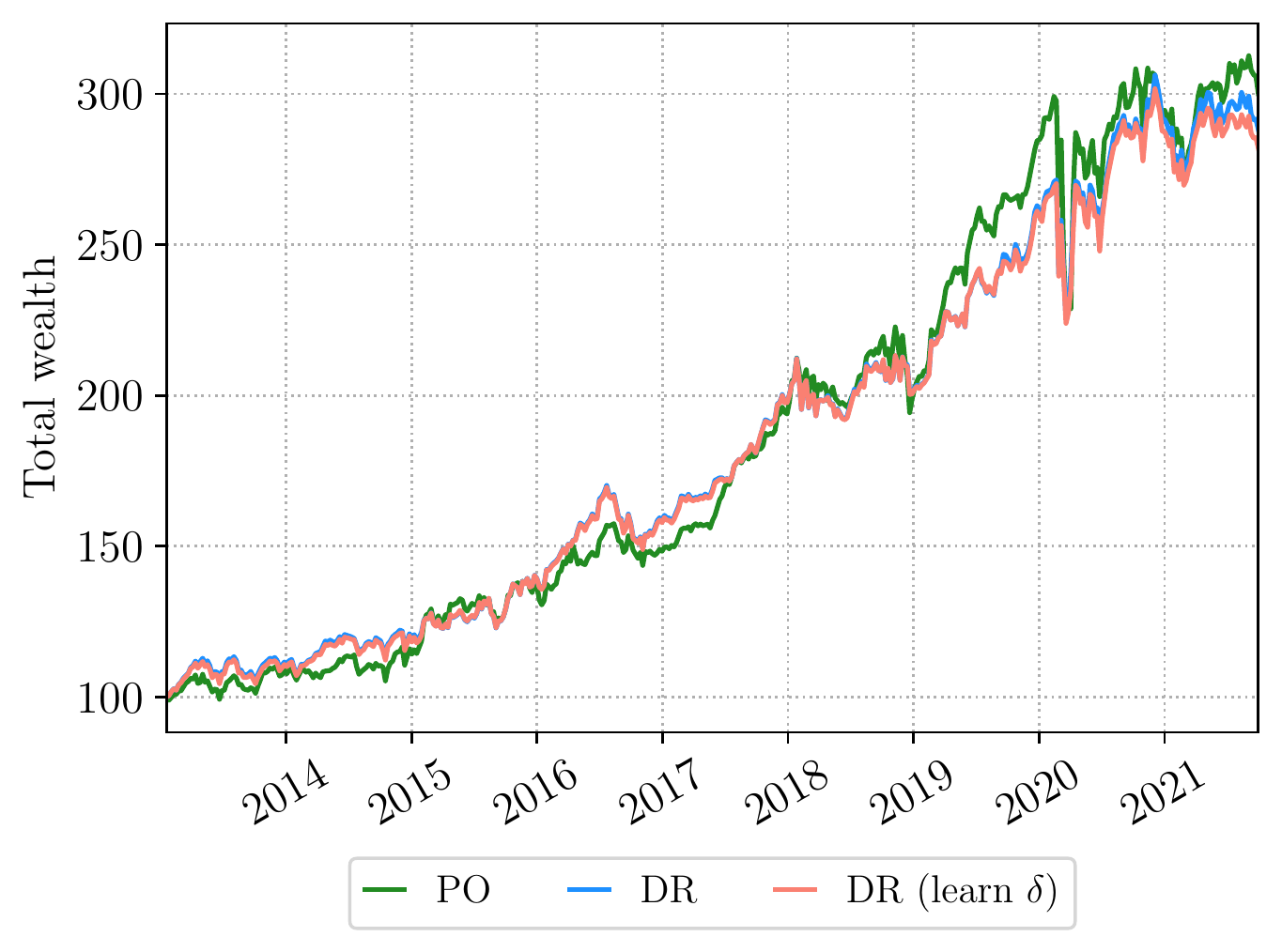}}
\caption{Experiment 2 -- Wealth evolution}
\label{fig:exp2_wealth}
\end{center}
\vskip -0.2in
\end{figure}

\begin{table}[t]
\caption{Experiment 2 -- Results}
\centering
\begin{tabular}{lrrr}
\toprule
				& PO & \multicolumn{2}{c}{DR}\\[0.5ex] \cline{3-4}
\rule{0pt}{3ex}&	 & Const. \(\delta\) & Learn \(\delta\)\\
\midrule
Return (\%)     & 13.4 & 12.9 & 12.6 \\
Volatility (\%) & 15.3 & 12.8 & 12.8 \\
Sharpe ratio    & 0.88 & 1.01 & 0.99 \\
\bottomrule
\multicolumn{4}{p{0.6\linewidth}}{\small\rule{0pt}{3ex}\textbf{Note}: Values are annualized.}\\
\end{tabular}
\label{table:exp2_results}
\end{table}

\subsubsection{Experiment 3: Learning \texorpdfstring{\(\gamma\)}{Lg}}\label{sec:learn_gamma}
The third experiment assessed the out-of-sample impact of learning 
\(\gamma\). Therefore, here we experimented with only those systems
involving \(\gamma\), i.e.  nominal system with constant \(\bm{\theta}\), 
and a DR system with constant \(\bm{\theta}\) and \(\delta\) but with a learnable
\(\gamma\), with the PO system added as a baseline control. Details of the three
systems are given in Table~\ref{table:exp3_models}.

\begin{table}[t]
\caption{Experiment 3 -- List of models}
\centering
\begin{tabular}{lrcr@{}rcr@{}rc}
\toprule
		& \multicolumn{2}{c}{\(\bm{\theta}\)} & & \multicolumn{2}{c}{\(\gamma\)} & & \multicolumn{2}{c}{\(\delta\)}\\[0.5ex] \cline{2-3} \cline{5-6} \cline{8-9}
\rule{0pt}{3ex}System 	& Val. 	& Lrn    && Val. & Lrn 	  && Val. & Lrn \\ 
\midrule
PO 		& OLS 	& -      &&	0.046 & -      && -     & -      \\[0.5ex]
Nominal & OLS 	& \xmark &&	0.046 & \cmark && -     & -      \\[0.5ex]
DR		& OLS 	& \xmark &&	0.046 & \cmark && 0.312 & \xmark \\[0.25ex]
\bottomrule
\multicolumn{9}{p{0.9\linewidth}}{\small\rule{0pt}{3ex}\textbf{Note(s)}: Val, Initial value; Lrn, Learnable.}\\
\end{tabular}
\label{table:exp3_models}
\end{table}

The out-of-sample financial performance of the three investment systems is 
summarized in Figure~\ref{fig:exp3_wealth} and Table~\ref{table:exp3_results}. 
The experimental results lead to the following observations.

\begin{itemize}[itemsep=-0.15em, topsep=0pt, leftmargin=*]

\item \emph{Impact of robustness}: The results of this experiment once again 
  highlight the out-of-sample benefits of incorporating robustness into the 
  system: the Sharpe ratio of the DR system is the highest. Although the PO
  system has the highest return, we must keep in mind that the task loss 
  function was a weighted combination of Sharpe ratio and prediction error, 
  meaning, by design, a higher Sharpe ratio is desirable.
  
\item \emph{Isolated learning of \(\gamma\)}: Comparing the PO system with the 
  nominal system, we observe that only learning \(\gamma\) may not be beneficial 
  to the out-of-sample performance of a system. However, when robustness is 
  incorporated into the system, the out-of-sample performance is greatly enhanced. 

\end{itemize}

\begin{figure}[ht]
\begin{center}
\centerline{\includegraphics[width=\columnwidth]{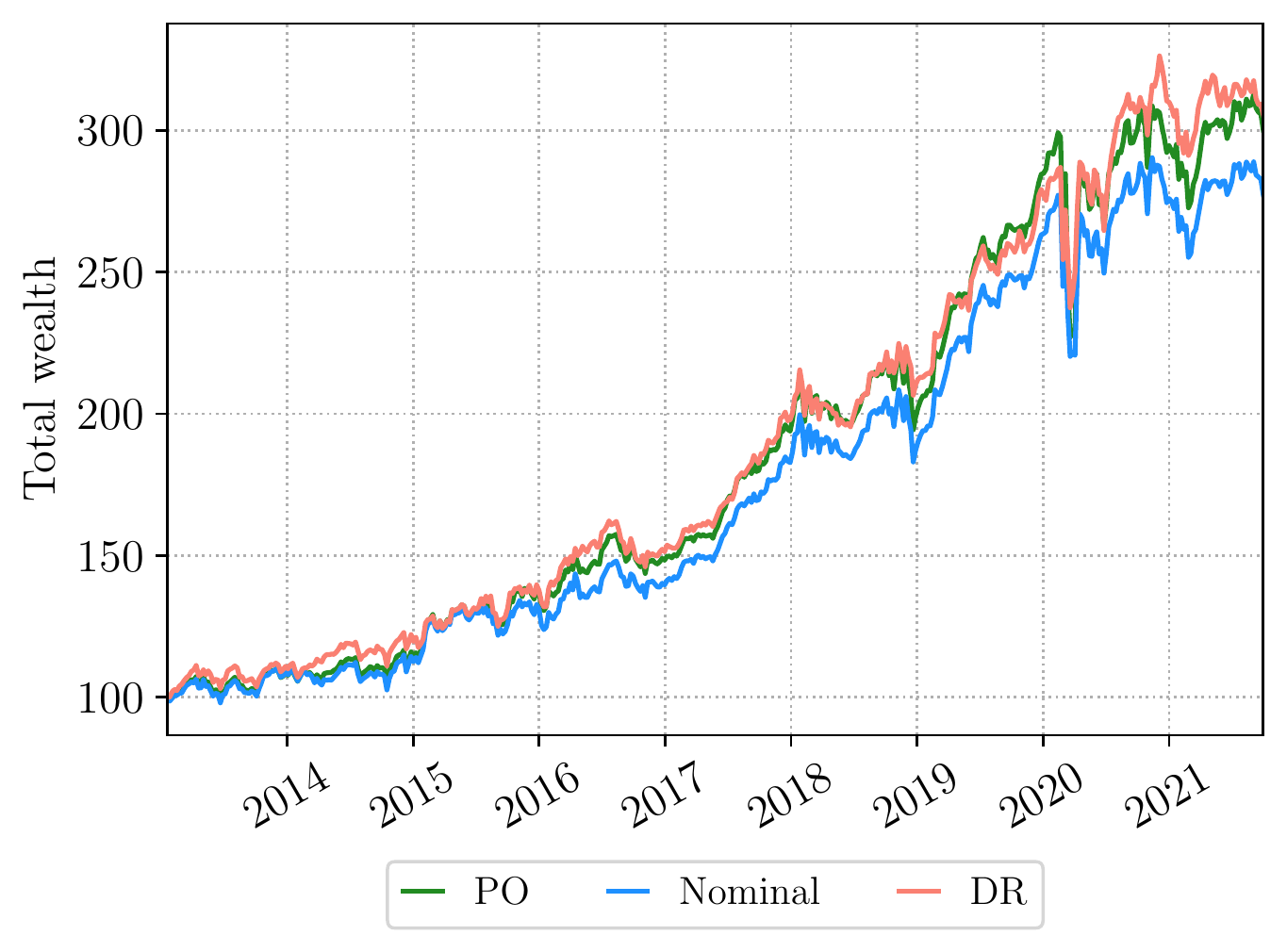}}
\caption{Experiment 3 -- Wealth evolution}
\label{fig:exp3_wealth}
\end{center}
\vskip -0.2in
\end{figure}

\begin{table}[t]
\caption{Experiment 3 -- Results}
\centering
\begin{tabular}{lrrr}
\toprule
& PO & Nom. & DR\\[0.5ex]
\midrule
Return (\%)     & 13.4 & 12.4 & 13.6\\
Volatility (\%) & 15.3 & 14.7 & 13.4\\
Sharpe ratio    & 0.88 & 0.84 & 1.02\\
\bottomrule
\multicolumn{4}{p{0.6\linewidth}}{\small\rule{0pt}{3ex}\textbf{Note(s)}:
  Values are annualized.}\\
\end{tabular}
\label{table:exp3_results}
\end{table}

\subsubsection{Experiment 4: Learning \texorpdfstring{\(\theta\)}{Lg}}\label{sec:learn_theta}

The fourth experiment assesses the out-of-sample impact of learning the 
prediction layer weights \(\bm{\theta}\). We tested four investment systems:
the PO system, the base system, the nominal system with constant \(\gamma\), 
and the DR system with constant \(\gamma\) and \(\delta\). Details of the 
four systems are presented in Table~\ref{table:exp4_models}.

\begin{table}[t]
\caption{Experiment 4 -- List of models}
\centering
\begin{tabular}{lrcr@{}rcr@{}rc}
\toprule
		& \multicolumn{2}{c}{\(\bm{\theta}\)} & & \multicolumn{2}{c}{\(\gamma\)} & & \multicolumn{2}{c}{\(\delta\)}\\[0.5ex] \cline{2-3} \cline{5-6} \cline{8-9}
\rule{0pt}{3ex}System 	& Val. 	& Lrn    && Val. & Lrn 	  && Val. & Lrn \\
\midrule
PO 		& OLS 	& -      &&	0.046 & -      && -     & -      \\[0.5ex]
Base 	& OLS 	& \cmark &&	-	  & -      && -     & -      \\[0.5ex]
Nominal	& OLS 	& \cmark &&	0.046 & \xmark && - 	& -      \\[0.5ex]
DR		& OLS 	& \cmark &&	0.046 & \xmark && 0.312 & \xmark \\[0.25ex]
\bottomrule
\multicolumn{9}{p{0.9\linewidth}}{\small\rule{0pt}{3ex}\textbf{Note}: Val, Initial value; Lrn, Learnable.}\\
\end{tabular}
\label{table:exp4_models}
\end{table}

The out-of-sample financial performance of the four investment systems is 
presented in Figure~\ref{fig:exp4_wealth} and Table~\ref{table:exp4_results}. 
The experimental results lead to the following observations.

\begin{itemize}[itemsep=-0.15em, topsep=0pt, leftmargin=*]
  
\item \emph{Impact of incorporating a deviation risk measure}: The
  difference in the out-of-sample performance between the base system and 
  the PO system clearly highlights the benefit of the deviation risk 
  measure -- even though the base system learns \(\bm{\theta}\), this 
  is not enough to combat the inherent uncertainty of the portfolio 
  returns.

\item \emph{Learning \(\bm{\theta}\)}: Learning \(\theta\) becomes 
  advantageous once a risk measure is added to the decision layer. The 
  PO and nominal systems differ only in that the nominal system is able 
  to learn values of \(\bm{\theta}\) that differ from the OLS weights.
  Thus, as suggested by \citet{donti2017task}, learning \(\bm{\theta}\) 
  enhances the mapping of the prediction layer from the feature space to 
  the asset space to extract a higher quality decision, but only if the 
  impact of prediction error in the decision layer is properly modeled.

\item \emph{Impact of robustness}: Comparing the nominal and DR systems, 
  we can see that incorporating robustness may not always be advantageous, in 
  particular when the robustness sizing parameter \(\delta\) is not optimally 
  calibrated. The results in Table~\ref{table:exp4_results} indicate that 
  \(\delta\) was set too conservative -- even though the volatility of the DR 
  system is lower than the nominal, it incurs a large opportunity cost with 
  regards to the portfolio return.
  
\end{itemize}

\begin{figure}[ht]
\begin{center}
\centerline{\includegraphics[width=\columnwidth]{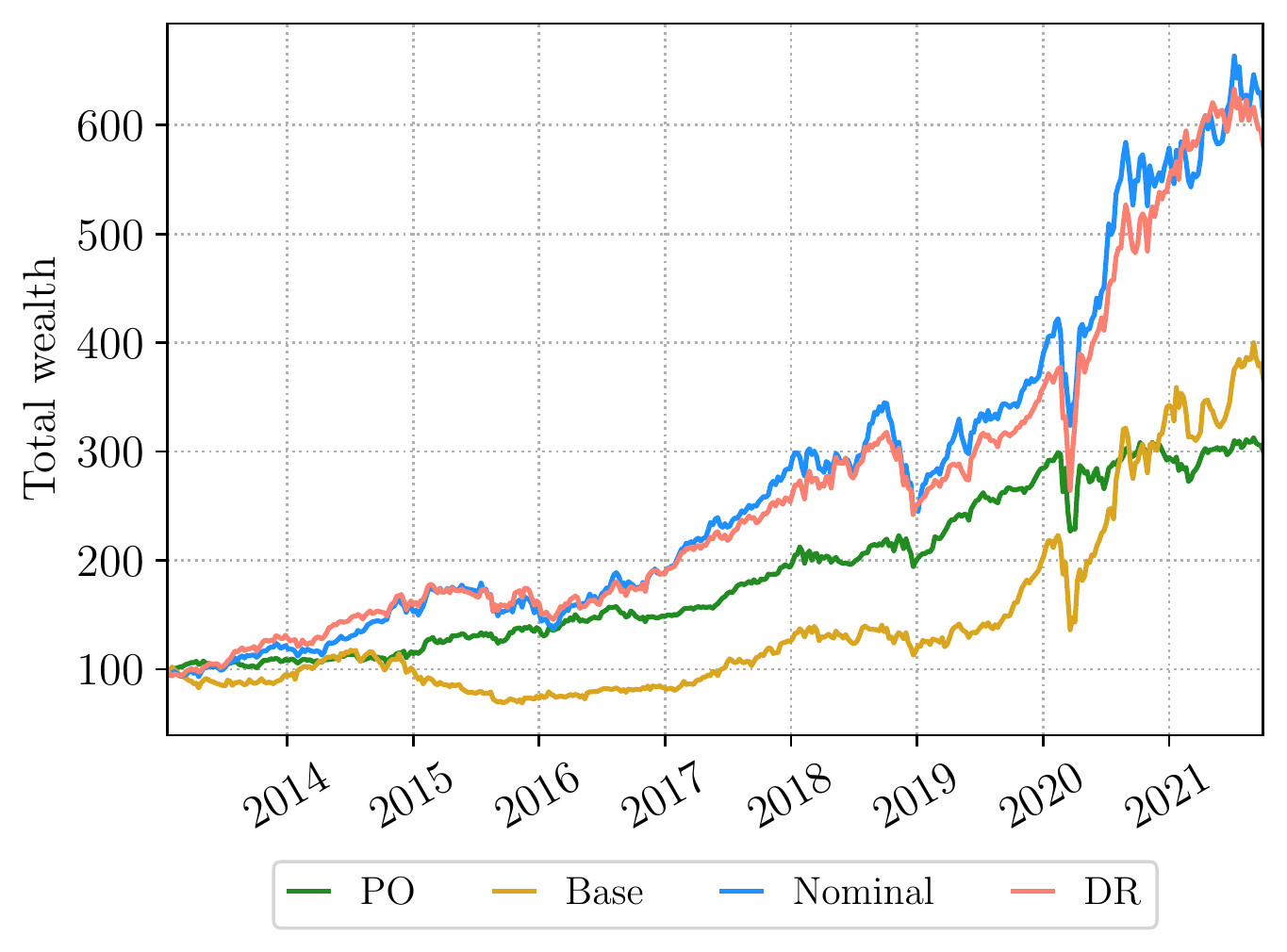}}
\caption{Experiment 4 -- Wealth evolution}
\label{fig:exp4_wealth}
\end{center}
\vskip -0.2in
\end{figure}

\begin{table}[t]
\caption{Experiment 4 -- Results}
\centering
\begin{tabular}{lrrrr}
\toprule
				& PO & Base & Nom. & DR \\
\midrule
Return (\%)     & 13.4 & 16.1 & 23.0 & 22.3 \\
Volatility (\%) & 15.3 & 25.0 & 19.1 & 18.5 \\
Sharpe ratio    & 0.88 & 0.64 & 1.21 & 1.21 \\
\bottomrule
\multicolumn{5}{p{0.6\linewidth}}{\small\rule{0pt}{3ex}\textbf{Note(s)}: Values are annualized.}\\
\end{tabular}
\label{table:exp4_results}
\end{table}

\subsection{Experiment with synthetic data}\label{sec:exp_synt}

We simulate an investment universe with 10~assets and 5~features. The features 
were assumed to follow an uncorrelated zero-mean Brownian motion. The asset returns 
are generated from a linear model of the features,
\[
    \bm{y}_t = \bm{\alpha} + \bm{\beta}^\top \bm{x}_t + \bm{\xi}_t + \kappa_t \cdot \bm{\omega}_t,
\]
where \(\bm{\alpha}\sim\mathcal{U}(0, 0.015)\in\mathbb{R}^n\) is the vector of biases,
\(\bm{\beta}\in\mathbb{R}^{m\times n}\) where \(\beta_{ij}\sim\mathcal{N}(0, 0.015)\ 
\forall\ i,j\) is the matrix of weights (i.e., loadings), \(\bm{\xi}_t\in\mathbb{R}^n\) 
where \(\xi_t^i\sim\mathcal{N}(0,0.015)\ \forall\ i\) is a Gaussian noise vector, 
\(\bm{\omega}_t\in\mathbb{R}^n\) where \(\omega_t^i\sim\mathrm{Exp}(0.015)\ \forall\ i\) 
is an exponential noise vector, and \(\kappa_t\) is a discrete random variable that 
takes the value \(-1,0,1\) with probability 0.15, 0.7, 0.15, respectively, and serves 
to periodically perturb the asset return process bidirectionally with the magnitude of 
\(\bm{\omega}_t\) in order to simulate `jumps' in the asset returns.

The experimental data set is 
composed of 1,200 observations, with the first 840 reserved for training and 
the remaining 360 reserved for testing. However, unlike previous experiments, 
the synthetic data allows us to use a single fold for validation. Here, the 
complete training set is separated into a single training subset and single 
validation subset. The validation results are shown in Table~\ref{table:exp5_val} 
in Appendix~\ref{app:add_exp}. The table also highlights the optimal 
hyperparameters selected for each system. 

The fifth experiment explored the advantage of robustness when the prediction 
layer is more complex, e.g. a neural network with multiple hidden layers. Specifically, 
this experiment compares nominal and DR systems when the prediction layer is either 
a linear model, or neural network with two or three hidden layers. The neural networks 
had fully connected hidden layers, and the activation functions were rectified linear 
units (ReLUs).

The prediction layer in this experiment is initialized to random weights using the 
standard PyTorch~\citep{paszke2019pytorch} initialization mechanism. This applies 
to all three prediction layer designs tested in this experiment. The initial values 
of the risk appetite parameter \(\gamma\) and robustness parameter \(\delta\) were 
sampled uniformly from the same intervals as the previous experiments (see 
Appendix~\ref{app:init} for the initialization methodology).

The objective of the experiment is to investigate whether robustness enhances 
the out-of-sample performance of a system with a more complex prediction 
layer, i.e., with a prediction layer that is more difficult to train accurately. 
To avoid biases pertaining to the design of the prediction layer, our assessment 
is based on a pairwise comparison between the two and three-layer systems. 
Details of the four investment systems are presented in 
Table~\ref{table:exp5_models}.

\begin{table}[t]
\caption{Experiment 5 -- List of models}
\centering
\begin{tabular}{lccr@{}rcr@{}rc}
\toprule
& \multicolumn{2}{c}{\(\bm{\theta}\)} && \multicolumn{2}{c}{\(\gamma\)} && \multicolumn{2}{c}{\(\delta\)}\\[0.5ex] \cline{2-3} \cline{5-6} \cline{8-9}
\rule{0pt}{3ex}System 	& Val. 	& Lrn    && Val. & Lrn 	  && Val. & Lrn \\
\midrule
\multicolumn{2}{l}{\textbf{Linear}}\\[0.5ex]
Nom. & Note 1 & \cmark && 0.089 & \cmark && -     & -      \\[0.5ex]
DR 	 & Note 1 & \cmark && 0.089 & \cmark && 0.146 & \cmark \\[2ex]
\multicolumn{2}{l}{\textbf{2-layer}}\\[0.5ex]
Nom. & Note 1 & \cmark && 0.081 & \cmark && -     & -      \\[0.5ex]
DR 	 & Note 1 & \cmark && 0.081 & \cmark && 0.358 & \cmark \\[2ex]
\multicolumn{2}{l}{\textbf{3-layer}}\\[0.5ex]
Nom. & Note 1 & \cmark && 0.072 & \cmark && -     & -      \\[0.5ex]
DR	 & Note 1 & \cmark && 0.072 & \cmark && 0.163 & \cmark \\[0.25ex]
\bottomrule
\multicolumn{9}{p{0.95\linewidth}}{\small\rule{0pt}{3ex}\textbf{Notes}: (1) \(\bm{\theta}\) is initialized to the same values for each pair of systems. (2) Val, Initial value; Lrn, Learnable.}\\
\end{tabular}
\label{table:exp5_models}
\end{table}

The out-of-sample financial performance of the four investment systems is 
summarized in Figure~\ref{fig:exp5_wealth} and Table~\ref{table:exp5_results}. 
As previously noted, our assessment is based on a pairwise comparison between 
the linear, and two- and three-layer systems. The experimental results 
lead to the following observation.

\begin{itemize}[itemsep=-0.15em, topsep=0pt, leftmargin=*]
  
\item \emph{Impact of robustness}: As indicated by the Sharpe ratios in
  Table~\ref{table:exp5_results}, introducing distributional robustness  
  greatly enhances the portfolio out-of-sample performance in all three cases. 
  Recall that, unlike previous experiments where the prediction layers were
  initialized to the naturally intuitive OLS weights, the prediction layer in 
  the current systems are initialized to fully randomized weights. Thus, 
  through this experiment we can appreciate how robustness protects the 
  portfolios from model error, particularly as the complexity of the 
  prediction layer increases. 
    
\end{itemize}

\begin{figure*}[ht]
\begin{center}
\centerline{\includegraphics[width=\textwidth]{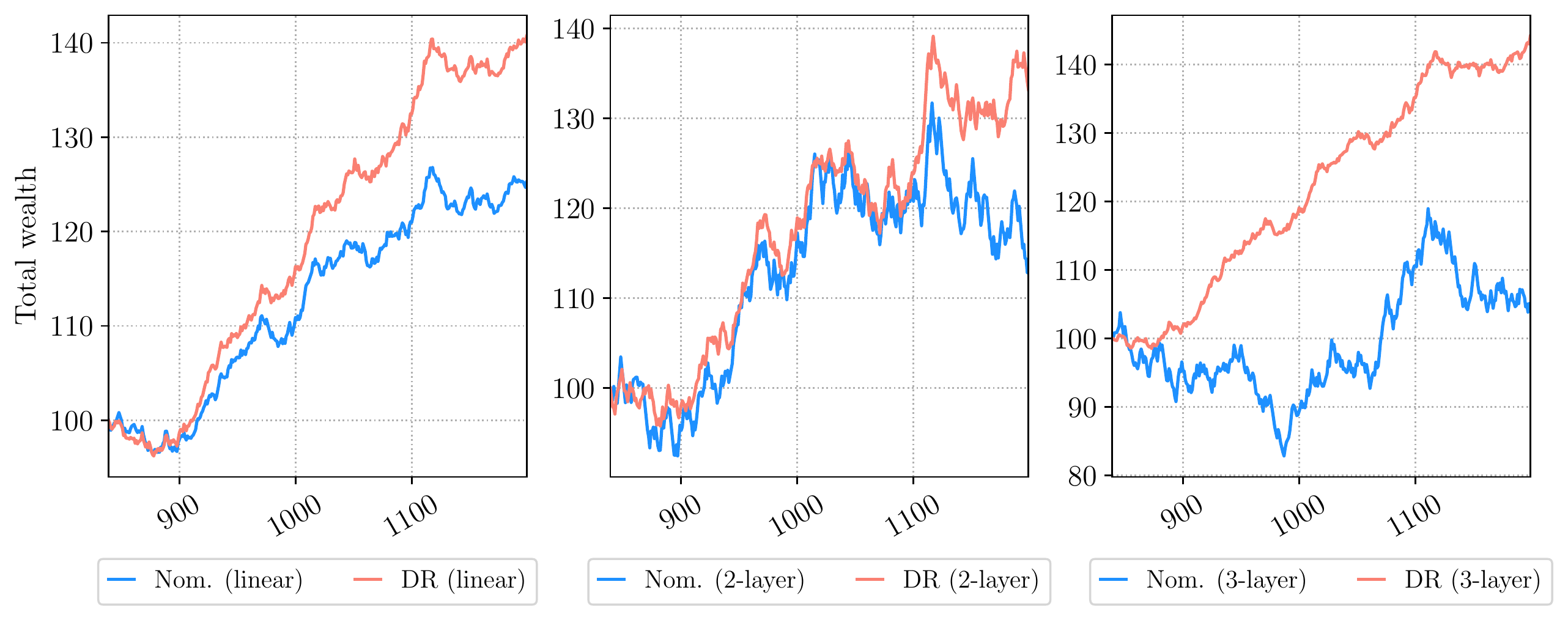}}
\caption{Experiment 5 -- Wealth evolution}
\label{fig:exp5_wealth}
\end{center}
\vskip -0.2in
\end{figure*}

\begin{table}[t]
\caption{Experiment 5 -- Results}
\centering
\begin{tabular}{lrrr@{}rrr@{}rr}
\toprule
				& \multicolumn{2}{c}{Linear} && \multicolumn{2}{c}{2-layer} && \multicolumn{2}{c}{3-layer}\\[0.5ex] \cline{2-3} \cline{5-6} \cline{8-9}
\rule{0pt}{3ex}& Nom. & DR && Nom. & DR && Nom. & DR \\
\midrule
Return (\%) & 3.30 & 5.10 && 1.80 & 4.20 && 0.70 & 5.40 \\
Vol. (\%) 	& 2.90 & 2.70 && 8.30 & 5.80 && 7.80 & 2.60 \\
SR    		& 1.16 & 1.88 && 0.21 & 0.73 && 0.08 & 2.11 \\
\bottomrule
\multicolumn{9}{p{0.99\linewidth}}{\small\rule{0pt}{3ex}\textbf{Notes}: (1) Vol, Volatility; SR, Sharpe ratio. (2) Values are annualized.}\\
\end{tabular}
\label{table:exp5_results}
\end{table}

\section{Conclusion}\label{sec:conclusion}

This paper introduces a novel DR end-to-end learning system for 
portfolio construction. Specifically, the system integrates a prediction 
layer of asset returns with a decision layer that selects portfolios by 
solving a DR optimization problem that explicitly incorporates model 
risk.

The decision layer in our end-to-end system consists of a DR portfolio 
selection problem that gets as input both the point prediction and the 
prediction errors. The prediction errors are used to define a deviation 
risk measure, and we penalize the performance of a portfolio with the 
worst-case risk over a set of probability measures. We show that, even 
though the deviation risk measure is not a linear function of the probability 
measure, one can still use convex duality to reformulate the minimax DR 
optimization problem into an equivalent minimization problem. This 
reformulation is critical to ensure that the gradients with respect to the 
task loss can be back-propagated to the prediction layer. Numerical experiments 
clearly show that incorporating model robustness via a DR layer leads to enhanced 
financial performance.

The parameters that control risk and robustness are learnable as part of
the end-to-end system, meaning these two parameters are optimized directly
from data based on the system's out-of-sample performance. The learnability 
of the robustness parameter is a consequence of the convex duality that we 
use to reformulate the DR layer. Setting these two parameters in
practice is a challenge, and our end-to-end system relieves the user from
having to set them. Furthermore, our numerical experiments show that these
parameters do significantly impact performance.

We have implemented our proposed robust end-to-end system in Python, and
have made the code available on GitHub. We anticipate the DR approach 
to be very impactful for any application where model risk is an important
consideration. Portfolio construction is only the beginning.

\bibliography{MyBib}
\bibliographystyle{apalike}

\begin{appendices}
\section{Proof of Proposition \ref{prop:dev_risk}}\label{app:dev_risk}

We begin with some preliminary information. Recall that \(\bm{\epsilon} = 
\{\bm{\epsilon}_j\in\mathbb{R}^n: j = 1, \ldots, T\}\) denotes the finite 
set of prediction error outcomes and \(\bm{p} \in \mathcal{Q}\) is a PMF.

Here we prove the four properties of the deviation risk measure 
\(f_{\bm{\epsilon}}\) outlined in Proposition \ref{prop:dev_risk}.

\begin{enumerate}[itemsep=-0.15em, topsep=0pt, leftmargin=*]

\item \(f_{\bm{\epsilon}}(\cdot, \bm{p}): \mathcal{Z} \rightarrow \mathbb{R}_+\) 
  is convex for any fixed \(\bm{p}\in\mathcal{Q}\).

  \begin{proof}
  Since \(R(X)\) is convex, it follows that 
  \begin{equation}
  \label{eq:h_ep}
	h_{\bm{\epsilon}}(c,\bm{z},\bm{p}) \triangleq \sum_{j=1}^T p_j\cdot R(\bm{\epsilon}_j^\top\bm{z} - c)
  \end{equation}
  is a convex function in \((c,\bm{z})\) for a fixed \(\bm{p}\). Accordingly,  
  \(f_{\bm{\epsilon}}(\bm{z}, \bm{p}) = \min_{c} h_{\bm{\epsilon}}(c,\bm{z},\bm{p})\) 
  is also a convex function of \(\bm{z}\) for fixed \(\bm{p}\)~\citep{boyd2004convex}. 
  \end{proof} 
	
  \item \(f_{\bm{\epsilon}}(\bm{z}, \bm{p})\geq 0\) for all 
    \(\bm{p}\in\mathcal{Q}\) and \(\bm{z}\in\mathcal{Z}\).
	
  \begin{proof} 
  By definition, we have that \(R\) is a non-negative function. Moreover, 
  \(\bm{p}\) is constrained to the probability simplex (specifically, 
  \(p_j\geq 0\ \forall\ j\)). Since, \(f_{\bm{\epsilon}}\) is the sum of
  \(T\) functions, each of which results from the product of two non-negative 
  elements, then \(f_{\bm{\epsilon}}(\bm{z}, \bm{p})\geq 0\). 
  \end{proof}
	
  \item \(f_{\bm{\epsilon}}(\bm{z}, \bm{p})\) is shift-invariant with respect 
  to \(\bm{\epsilon}\), i.e. \(f_{\bm{\epsilon}+\bm{a}}(\bm{z}, \bm{p}) = 
  f_{\bm{\epsilon}}(\bm{z}, \bm{p}) \) for any fixed vector \(\bm{a} \in 
  \mathbb{R}^n\).
  \begin{proof} 
  By definition, we have that
  \begin{align*}
  	f_{\bm{\epsilon}+\bm{a}}(\bm{z}, \bm{p}) 
	&= \min_c \bigg\{\sum_{j=1}^T p_j\cdot R\big((\bm{\epsilon}_j 
		+ \bm{a})^\top \bm{z} - c\big)\bigg\},\\
	&= \min_c \bigg\{\sum_{j=1}^T p_j\cdot R\big(\bm{\epsilon}_j^\top
		\bm{z} - (c - \bm{a}^\top\bm{z})\big)\bigg\},\\
	&= \min_{c'} \bigg\{\sum_{j=1}^T p_j\cdot R\big(\bm{\epsilon}_j^\top
		\bm{z} - c')\big)\bigg\},
  \end{align*}
  where \(c' = c - \bm{a}^\top \bm{z}\). The result follows from the last 
  expression.
  \end{proof}
		
  \item \(f_{\bm{\epsilon}}(\bm{z}, \bm{p})\) is symmetric with respect 
  to \(\bm{\epsilon}\), i.e. \(f_{\bm{\epsilon}}(\bm{z}, \bm{p}) = 
  f_{-\bm{\epsilon}}(\bm{z}, \bm{p}) \).
  \begin{proof}
  By definition, we have that
  \begin{align}
	f_{-\bm{\epsilon}}(\bm{z}, \bm{p}) 
	&= \min_c \bigg\{\sum_{j=1}^T p_j\cdot R\big(-\bm{\epsilon}_j^\top \bm{z} 
	- c\big)\bigg\},\nonumber \\
	&= \min_c \bigg\{\sum_{j=1}^T p_j\cdot R\big(\bm{\epsilon}_j^\top \bm{z} 
	+ c \big)\bigg\}, \label{eq:Rsym}\\
	&= \min_{c'} \bigg\{\sum_{j=1}^T p_j\cdot R\big(\bm{\epsilon}_j^\top \bm{z} 
	- c')\big)\bigg\}, \nonumber
  \end{align}
  where \(c' = - c \), and \eqref{eq:Rsym}  follows from the fact that 
  \(R(X) = R(-X)\). The result follows from the last expression.
  \end{proof}
          		
\end{enumerate}

\section{Portfolio error variance}\label{app:variance}

Recall that \(\bm{\epsilon} = \{\bm{\epsilon}_j\in\mathbb{R}^n: j = 1, 
\ldots, T\}\) denotes the finite set of prediction error outcomes. For 
a fixed distribution \(\bm{p}\in\mathcal{Q}\), the expected error and 
its corresponding covariance matrix are 
\begin{align}
	\bm{\mu}_{\bm{\epsilon}}(\bm{p}) 
		&\triangleq \sum_{j=1}^T p_j \cdot \bm{\epsilon}_j,\label{eq:mu}\\ 
	\bm{\Sigma}_{\bm{\epsilon}}(\bm{p}) 
		&\triangleq \sum_{j=1}^T p_j\cdot \big(\bm{\epsilon}_j - 
			\bm{\mu}_{\bm{\epsilon}}(\bm{p})\big)\big(\bm{\epsilon}_j 
			- \bm{\mu}_{\bm{\epsilon}}(\bm{p})\big)^\top,\label{eq:Sigma} 
\end{align}
where \(\bm{\mu}_{\bm{\epsilon}}(\bm{p})\in\mathbb{R}^n\) and
\(\bm{\Sigma}_{\bm{\epsilon}}(\bm{p})\in\mathbb{R}^{n\times n}\). The
matrix \(\bm{\Sigma}_{\bm{\epsilon}}(\bm{p})\) results from the weighted
sum of \(T\) rank-1 symmetric matrices, meaning it is guaranteed to be
positive semidefinite. 

Note that \(\bm{\Sigma}_{\bm{\epsilon}}(\bm{p})\) is a non-linear function
of \(\bm{p}\); consequently, the portfolio variance \(\bm{z}^\top
\hat{\bm{\Sigma}}_{\bm{\epsilon}}(\bm{p}) \bm{z}\) is also non-linear in
\(\bm{p}\).  Thus, in this form, the worst-case variance
\(\max_{\bm{p} \in \mathcal{P}(\delta)} \bm{z}^\top
\hat{\bm{\Sigma}}_{\bm{\epsilon}}(\bm{p}) \bm{z}\) does not have the
correct convexity properties that allow one to reformulate the problem
using duality. We resolve this issue by recasting the portfolio error
variance into the form prescribed by Proposition \ref{prop:dev_risk},
which yields the following corollary.  

\begin{corollary}
\label{coroll:variance}
	For a fixed \(\bm{z}\in\mathcal{Z}\) and \(\bm{p}\in\mathcal{Q}\), 
	the portfolio variance is 
	\[
		\bm{z}^\top \bm{\Sigma}_{\bm{\epsilon}}(\bm{p}) \bm{z} =
                \min_{c}\ \sum_{j=1}^T p_j\cdot (\bm{\epsilon}_j^\top
                \bm{z} - c)^2, 
	\]
	where \(c\in\mathbb{R}\) is an unrestricted auxiliary variable.
\end{corollary}
\begin{proof}
	Let \(h(c) = \sum_{j=1}^T p_j\cdot (\bm{\epsilon}_j^\top \bm{z} -
        c)^2\) and recall that \(\sum_{j=1}^T p_j = 1\). Then 
	\begin{align}
		\frac{dh}{dc} &= 2 \sum_{j=1}^T (p_j\cdot c - p_j\cdot \bm{\epsilon}_j^\top \bm{z})\nonumber\\
			&= 2c - 2\sum_{j=1}^T p_j\cdot \bm{\epsilon}_j^\top \bm{z},\nonumber\\
		\frac{d^2 h}{dc^2} &= 2.\nonumber
	\end{align}
	Since \(d^2 h/dc^2 > 0\ \forall c\in\mathbb{R}\), then \(h(c)\) is
        strictly convex in \(c\). Therefore, solving for \(c\) such that
        \(dh/dc = 0\) suffices to find the unique global minimizer
        \(c^*\),  
	\[
	\begin{aligned}
		\frac{dh}{dc}\bigg|_{c=c^*} &= 2c^* - 2\sum_{j=1}^T
                p_j\cdot \bm{\epsilon}_j^\top \bm{z} = 0,\\
		c^* &= \sum_{j=1}^T p_j\cdot \bm{\epsilon}_j^\top \bm{z},
	\end{aligned}
	\]
	for all \(\bm{p}\in\mathcal{Q}\) and \(\{\bm{\epsilon}_j: j = 1, \ldots, T\}\). Then, by \eqref{eq:mu}, we have
	\[
		c^* = \sum_{j=1}^T p_j\cdot \bm{\epsilon}_j^\top \bm{z} =
                \bm{z}^\top \sum_{j=1}^T p_j\cdot \bm{\epsilon}_j =
                \bm{z}^\top \bm{\mu}_{\bm{\epsilon}}(\bm{p}).
	\]
	Now, using \eqref{eq:Sigma}, the portfolio error variance is
	\begin{align*}
          \bm{z}^\top \bm{\Sigma}_{\bm{\epsilon}}(\bm{p})
          \bm{z} &= \sum_{j=1}^T p_j\cdot
                   \bm{z}^\top\big(\bm{\epsilon}_j -
                   \bm{\mu}_{\bm{\epsilon}}(\bm{p})\big)\big(\bm{\epsilon}_j
                   - \bm{\mu}_{\bm{\epsilon}}(\bm{p})\big)^\top
                   \bm{z}\\ 
                 &= \sum_{j=1}^T p_j\cdot
                   \Big(\bm{z}^\top\big(\bm{\epsilon}_j -
                   \bm{\mu}_{\bm{\epsilon}}(\bm{p})\big)\Big)^2\\ 
                 &= \sum_{j=1}^T p_j\cdot \big(\bm{z}^\top
                   \bm{\epsilon}_j - \bm{z}^\top
                   \bm{\mu}_{\bm{\epsilon}}(\bm{p})\big)^2\\ 
                 &= \sum_{j=1}^T p_j\cdot (\bm{\epsilon}_j^\top \bm{z} - c^*)^2\\
			&= \min_{c} \sum_{j=1}^T p_j\cdot (\bm{\epsilon}_j^\top \bm{z} - c)^2,
	\end{align*}
	as desired. 
\end{proof} 

\section{Dualizing the DR layer}\label{app:dual}

Recall that \(h_{\bm{\epsilon}}(c,\bm{z},\bm{p})\) in \eqref{eq:h_ep} is
a convex function in \((c, \bm{z})\) for any fixed 
\(\bm{p}\in\mathcal{Q}\). Moreover, \(h_{\bm{\epsilon}}(c,\bm{z},\bm{p})\) 
is linear in \(\bm{p}\) for any fixed 
\((c,\bm{z})\in\mathbb{R}\times\mathcal{Z}\). Thus,
\(h_{\bm{\epsilon}}(c,\bm{z},\bm{p})\) is convex--linear in \((c,\bm{z})\)
and \(\bm{p}\), respectively.  

Fix \(\bm{z} \in \mathcal{Z}\). Then, the  minimax theorem for convex duality~\citep{neumann1928theorie} implies that 
\begin{align*}
	\max_{\bm{p}\in\mathcal{P}(\delta)} f_{\bm{\epsilon}}(\bm{z},\bm{p})
	&= \max_{\bm{p}\in\mathcal{P}(\delta)} \min_{c\in\mathbb{R}} h_{\bm{\epsilon}}(c,\bm{z},\bm{p})\\
	&=\min_{c\in\mathbb{R}} \max_{\bm{p}\in\mathcal{P}(\delta)} \sum_{j=1}^T p_j \cdot R(\bm{\epsilon}_j^\top \bm{z} - c).
\end{align*}

Next, we adapt the results in \citet{ben2013robust} to write the 
maximization over \(\bm{p}\) as a dual minimization problem. Note that 
this transformation is straightforward, and is only included for completeness.

Fix \((c,\bm{z})\in\mathbb{R}\times\mathcal{Z}\). From the definition of the ambiguity set
\(\mathcal{P}(\delta)\), it follows that the maximization 
problem in \(\bm{p}\) is
\[
\begin{array}{ll}
    \displaystyle \max_{\bm{p}} 
    &\displaystyle \sum_{j=1}^T p_j \cdot R(\bm{\epsilon}_j^\top \bm{z} - c),\\
    \ \text{s.t.} &\displaystyle \sum_{j=1}^T p_j = 1,\\
    &\displaystyle I_\phi(\bm{p},\bm{q}) =
      \sum_{j=1}^T q_j\cdot \phi(p_j/q_j) \leq \delta,\\
    & \bm{p} \geq \bm{0}.
\end{array}
\]
Associate a dual variable \(\xi\) with the constraint \(\sum_{j=1}^T
p_j = 1\) and a dual variable \(\lambda \geq 0\) with the constraint
\(I_\phi(\bm{p},\bm{q}) \leq \delta\). Then, the Lagrangian dual function is
\begin{align}
	&f_{\bm{\epsilon}}^{\delta}(\bm{z}, c, \lambda, \xi)\nonumber\\ 
	&\hspace{20pt} \triangleq \max_{\bm{p}\geq \bm{0}} \sum_{j=1}^T p_j \cdot R(\bm{\epsilon}_j^\top \bm{z} - c) + \xi\cdot (1 - \bm{1}^\top\bm{p}) \nonumber\\[-1ex]
	&\hspace{60pt} + \lambda\cdot\bigg(\delta - \sum_{j=1}^T q_j\cdot \phi(p_j/q_j)\bigg) \nonumber\\
	&\hspace{20pt}= \xi + \delta\cdot\lambda + \max_{\bm{p}\geq \bm{0}} \sum_{j=1}^T \Big(p_j\cdot \big(R(\bm{\epsilon}_j^\top \bm{z} - c) - \xi\big)\nonumber\\[-1ex]
	&\hspace{135pt} - \lambda\cdot q_j\cdot \phi(p_j/q_j)\Big)\nonumber\\
	&\hspace{20pt}= \xi + \delta\cdot\lambda + \sum_{j=1}^T \max_{p_j\geq 0} \Big(p_j\cdot \big(R(\bm{\epsilon}_j^\top \bm{z} - c) - \xi\big)\nonumber\\[-1ex]
	&\hspace{135pt} - \lambda\cdot q_j\cdot \phi(p_j/q_j)\Big)\nonumber\\
	&\hspace{20pt}= \xi + \delta\cdot\lambda + \sum_{j=1}^T q_j\cdot \max_{s\geq 0} \Big(s\cdot \big(R(\bm{\epsilon}_j^\top \bm{z} - c) - \xi\big)\nonumber\\[-1ex]
	&\hspace{178pt} - \lambda\cdot \phi(s)\Big)\nonumber\\
	&\hspace{20pt}= \xi + \delta\cdot\lambda + \sum_{j=1}^T q_j\cdot (\lambda \phi)^*\big(R(\bm{\epsilon}_j^\top \bm{z} - c) - \xi\big)\nonumber\\
	&\hspace{20pt}=\xi + \delta\cdot \lambda + \frac{\lambda}{T}
   \sum_{j=1}^T \phi^*\bigg(\frac{R(\bm{\epsilon}_j^\top \bm{z} -
   c) - \xi}{\lambda}\bigg).\label{eq:dual}
\end{align}
Note that we arrive at \eqref{eq:dual} by taking the convex conjugate of
\(\phi\) and by using the identity \((\lambda \phi)^*(w) = \lambda\cdot
\phi^*(w/\lambda)\) for \(\lambda\geq 0\). Additionally, recall that our
nominal assumption states that \(q_j = 1/T\ \forall j\).  

Since the conjugate \(\phi^*\) of a \(\phi\)-divergence is a convex
function, it follows that the Lagrangian function 
\(f_{\bm{\epsilon}}^{\delta}(\bm{z}, c, \lambda, \xi)\) is \emph{jointly} 
convex in \((\bm{z},c,\lambda,\xi) \in \mathcal{Z} \times \mathbb{R} \times
\mathbb{R}_+ \times \mathbb{R}\) (see \citet{ben2013robust} for details).

Finally, note that the Lagrangian dual in \eqref{eq:dual} above is the same
as the distributionally robust deviation risk measure in \eqref{eq:dr_obj}. 
The above steps demonstrate the equivalence between the minimax problem in 
\eqref{eq:dr_minimax} and the minimization problem in \eqref{eq:dr_opt}, i.e.,
\begin{align*}
	&\min_{\bm{z}\in\mathcal{Z}, c}\ \max_{\bm{p}\in\mathcal{P}(\delta)} 
		\sum_{j=1}^T p_j \cdot R(\bm{\epsilon}_j^\top \bm{z} - c) 
		- \gamma \cdot \hat{\bm{y}}_{t}^\top \bm{z}\\[1ex]
	&\hspace*{35pt} \iff 
		\min_{\bm{z}\in\mathcal{Z},\ \lambda\geq 0,\ \xi,\ c}\ 
		f_{\bm{\epsilon}}^{\delta}(\bm{z}, c, \lambda, \xi) 
		- \gamma \cdot \hat{\bm{y}}_{t}^\top \bm{z}.
\end{align*}

\section{Tractable reformulations}\label{app:reform}

Recall that the DR layer in \eqref{eq:dr_opt} corresponds to the 
following minimization problem,
\[
\begin{aligned}
	&\begin{aligned}
	&\min_{\bm{z}, \lambda, \xi, c} \xi + \delta\cdot \lambda + \frac{\lambda}{T} 
	\sum_{j=1}^T \phi^*\bigg(\frac{R(\bm{\epsilon}_j^\top \bm{z} - c) - 
	\xi}{\lambda}\bigg) - \gamma \cdot \hat{\bm{y}}_{t}^\top \bm{z}
	\end{aligned}\\
	&\begin{aligned}
	&\ \text{s.t.} & \lambda &\geq 0\\
	&& \bm{z} &\in\mathcal{Z}
	\end{aligned}	
\end{aligned}
\]
The computational tractability of this optimization problem depends on the 
complexity of the \(\phi\)-divergence selected to construct the ambiguity 
set \(\mathcal{P}(\delta)\).

\subsection{Hellinger distance}

The Hellinger distance is defined as 
\[
	I_\phi^H(\bm{p},\bm{q}) \triangleq \sum_{j=1}^T q_j\cdot \phi_H(p_j/q_j) =
	 \sum_{j=1}^T \big( \sqrt{p_j} - \sqrt{q_j} \big)^2
\]
where \(\phi_H(w) \triangleq (\sqrt{w} - 1)^2\) for \(w\geq 0\). The convex 
conjugate 
\[
	\phi_H^*(s) \triangleq \frac{s}{1-s}\ \text{ for } s < 1.
\]
If we construct the probability ambiguity set \(\mathcal{P}(\delta)\) 
using the Hellinger distance, then the DR layer becomes a highly non-linear 
convex optimization problem.
Nevertheless, the problem can be reformulated into a tractable problem as 
follows. Let \(h(s) \triangleq 1/(1-s) = \phi_H^*(s) + 1\). Then
\(h^{-1}(s) = 1-1/s\) and the objective function is
\[
	\xi + (\delta-1)\cdot \lambda + \frac{1}{T} 
	\sum_{j=1}^T \lambda\cdot h\bigg(\frac{R(\bm{\epsilon}_j^\top \bm{z} - c) - 
	\xi}{\lambda}\bigg) - \gamma \cdot \hat{\bm{y}}_{t}^\top \bm{z}.
\]
Next, introduce the auxiliary variable \(\bm{\beta}\in\mathbb{R}^T\) and let
\begin{align}
	&\beta_j \geq \lambda\cdot h\bigg(\frac{R(\bm{\epsilon}_j^\top \bm{z} - c) - 
	\xi}{\lambda}\bigg)\nonumber\\[1ex]
	&\hspace*{35pt} \iff h^{-1}\bigg(\frac{\beta_j}{\lambda}\bigg) \geq 
	\frac{R(\bm{\epsilon}_j^\top \bm{z} - c) - \xi}{\lambda}\nonumber\\[1ex]
	&\hspace*{35pt} \iff\hspace*{11pt} \lambda - \frac{\lambda^2}{\beta_j}\geq 
	R(\bm{\epsilon}_j^\top \bm{z} - c) - \xi.\nonumber
\end{align}
Note that, since \(\phi_H^*(s)\) requires that \(s<1\), then \(\bm{\beta} \geq \bm{0}\). 

Finally, introduce the auxiliary variable \(\bm{\tau}\in\mathbb{R}_+^T\) to 
arrive at the tractable reformulation of the Hellinger-based DR layer,
\[
\begin{aligned}
	&\begin{aligned}
	&\min_{\bm{z}, \lambda, \xi, c, \bm{\beta}, \bm{\tau}} \xi + 
	(\delta-1)\cdot \lambda + \frac{1}{T} \sum_{j=1}^T \beta_j  
	- \gamma \cdot \hat{\bm{y}}_{t}^\top \bm{z}
	\end{aligned}\\
	&\begin{aligned}
	&\quad\ \ \text{s.t.} & \xi + \lambda &\geq R(\bm{\epsilon}_j^\top \bm{z} - c) 
	+ \tau_j, & j=1,\dots,T\\
	&& \beta_j \tau_j &\geq \lambda^2, & j=1,\dots,T\\
	&& \bm{\beta}, \bm{\tau} &\geq \bm{0},\\
	&& \lambda &\geq 0,\\
	&& \bm{z} &\in\mathcal{Z}.
	\end{aligned}	
\end{aligned}
\]
Note that the hyperbolic constraint, \(\beta_j \tau_j \geq \lambda^2\), can be 
equivalently written as a rotated second-order cone constraint.

\subsection{Variation distance} 
The Variation distance is defined as 
\[
	I_\phi^V(\bm{p},\bm{q}) \triangleq \sum_{j=1}^T q_j\cdot \phi_V(p_j/q_j) =
	 \sum_{j=1}^T \big| p_j - q_j \big|
\]
where \(\phi_V(w) \triangleq |w - 1|\) for \(w\geq 0\). 
The
conjugate 

\[
	\phi_V^*(s) \triangleq \begin{cases}
 							-1,& s\leq -1\\
 							s,& -1\leq s\leq 1 
 							\end{cases}.
\]
The tractable reformulation presented here is adapted from~\citet{ben2013robust}. 
Introduce the auxiliary variable \(\bm{\beta}\in\mathbb{R}^T\). Then, 
\[
\begin{aligned}
	&\begin{aligned}
	&\min_{\bm{z}, \lambda, \xi, c, \bm{\beta}} \xi + \delta\cdot\lambda 
	+ \frac{1}{T} \sum_{j=1}^T \beta_j - \gamma\cdot\hat{\bm{y}}_{t}^\top \bm{z}
	\end{aligned}\\
	&\begin{aligned}
	&\ \text{s.t.} & \beta_j &\geq -\lambda, & j=1,\dots,T\\
	&& \beta_j &\geq R(\bm{\epsilon}_j^\top \bm{z} - c) - \xi, & j=1,\dots,T\\
	&& \lambda &\geq R(\bm{\epsilon}_j^\top \bm{z} - c) - \xi, & j=1,\dots,T\\
	&& \lambda &\geq 0,\\
	&& \bm{z} &\in\mathcal{Z}.
	\end{aligned}	
\end{aligned}
\]

\section{Initializing \texorpdfstring{\(\gamma\)}{Lg} and \texorpdfstring{\(\delta\)}{Lg} during experiments}\label{app:init}

The initial value for the parameter \(\gamma\) was randomly sampled 
from a uniform distribution over a finite interval. The upper and 
lower boundaries of this interval are set such that the nominal 
value of the deviation risk measure, i.e. portfolio error variance,
and the nominal return are comparable for the equal weight portfolio, 
i.e. \(\hat{z}_i = 1/n\) for \(i = 1,\ldots,n\). Note that the error 
variance was computed using errors corresponding to a linear prediction 
model with OLS weights.

Using the same training set as Section \ref{sec:num_ex} with a sample 
set of \(T=104\) prediction errors, we obtained samples of
plausible \(\gamma\) values as follows,
\[
	\hat{\gamma}_j = \frac{\hat{\bm{z}}^\top \bm{\Sigma}_{\bm{\epsilon}}(\bm{q}) \hat{\bm{z}}}{| \hat{\bm{y}}_j^\top \hat{\bm{z}} |} \quad \text{for}\ j=T+1,\ldots, T_0,
\]
where \(\bm{\Sigma}_{\bm{\epsilon}}(\bm{q})\) is calculated as in 
\eqref{eq:Sigma}. To ensure reasonable emphasis is given to the deviation 
risk measure, we set the lower and upper bounds for the interval for
\(\gamma\) to the 1st and 25th percentiles of the sample set, 
respectively. This gave us the interval \([0.02, 0.10]\). 

The initial value for \(\delta\) was also sampled from a uniform 
distribution over a finite interval. For the Hellinger distance,
which we use for our experiments, the maximum distance is 
\(\delta_{\max}  = \max_{p} I^{H}_{\phi}(p,q) = 2(1 - 1/\sqrt{T})\), 
where the nominal distribution \(\bm{q} = \frac{1}{T}\). 
The theoretical upper bound is a useful benchmark, but we note that this 
is an extreme and implausible value. Thus, we set the upper bound 
\(\delta_{\text{up}} = 0.25\cdot\delta_{\text{max}}\), while we set the lower 
bound \(\delta_{\text{lo}} = 0.05\cdot\delta_{\text{max}}\). In Experiments 
1--4, where \(T=104\), the sampling interval is \(\delta\in[0.09, 0.45]\). 
Initializing \(\delta\) from this interval ensures that the distance between 
\(\bm{p}\) and \(\bm{q}\) is not too extreme at the start of the training process.

\onecolumn
\section{Validation and hyperparameter selection}\label{app:add_exp}

\begin{table*}[!ht]
\caption{Average validation loss in Experiments 1--4}
\centering
\begin{tabular}{ccccccr@{}ccccc}
\toprule
		$\eta$ &  Epochs & Base & \multicolumn{3}{c}{Nominal} && \multicolumn{5}{c}{DR} \\[0.5ex] \cline{4-6} \cline{8-12}
&&& \multicolumn{3}{c}{\footnotesize Learned parameters} && \multicolumn{5}{c}{\footnotesize Learned parameters}\\[-0.5ex]
\rule{0pt}{3ex}&& & \(\bm{\theta}, \gamma\) & \(\gamma\) & \(\bm{\theta}\) && \(\bm{\theta}, \gamma, \delta\) & \(\delta\) & \(\gamma\) & \(\gamma, \delta\) & \(\bm{\theta}\) \\
\midrule
0.0050 & 30 & \textbf{-0.1522} & -0.1769 & -0.1682 & -0.1747 && -0.1708 & \textbf{-0.1897} & \textbf{-0.1841} & \textbf{-0.1817} & -0.1714 \\
0.0125 & 30 & -0.1130 & -0.1784 & -0.1672 & -0.1854 && -0.1723 & -0.1779 & -0.1775 & -0.1766 &  -0.1698 \\
0.0200 & 30 & -0.1097 & -0.1809 & -0.1675 & \textbf{-0.1874} && -0.1412 & -0.1785 & -0.1749 & -0.1766 &     -0.1603 \\
\midrule
0.0050 & 40 & -0.1372 & -0.1793 & -0.1680 & -0.1733 && -0.1705 & -0.1888 & -0.1817 & -0.1787 &     -0.1734 \\
0.0125 & 40 & -0.1109 & -0.1767 & -0.1675 & -0.1828 && -0.1703 & -0.1782 & -0.1762 & -0.1770 &  -0.1746 \\
0.0200 & 40 & -0.1044 & -0.1819 & -0.1675 & -0.1848 && -0.1365 & -0.1779 & -0.1752 & -0.1756 & -0.1609 \\
\midrule
0.0050 & 50 & -0.1278 & -0.1820 & -0.1679 & -0.1727 && -0.1748 & -0.1868 & -0.1807 & -0.1780 & -0.1719 \\
0.0125 & 50 & -0.1091 & -0.1772 & -0.1674 & -0.1839 && -0.1752 & -0.1780 & -0.1757 & -0.1757 & -0.1718 \\
0.0200 & 50 & -0.1079 & -0.1887 & -0.1683 & -0.1854 && -0.1329 & -0.1780 & -0.1755 & -0.1679 & -0.1649 \\
\midrule
0.0050 & 60 & -0.1267 & -0.1853 & -0.1679 & -0.1754 && -0.1805 & -0.1777 & -0.1798 & -0.1779 & -0.1689 \\
0.0125 & 60 & -0.1094 & -0.1807 & -0.1677 & -0.1846 && \textbf{-0.1826} & -0.1779 & -0.1761 & -0.1680 & -0.1653 \\
0.0200 & 60 & -0.0978 & \textbf{-0.1950} & \textbf{-0.1684} & -0.1838 && -0.1630 & -0.1781 & -0.1760 & -0.1674 & \textbf{-0.1756} \\
\midrule
0.0050 & 80 & -0.1069 & -0.1888 & -0.1679 & -0.1792 && -0.1779 & -0.1779 & -0.1788 & -0.1765 & -0.1692 \\
0.0125 & 80 & -0.1036 & -0.1804 & -0.1682 & -0.1846 && -0.1695 & -0.1779 & -0.1762 & -0.1675 & -0.1662 \\
0.0200 & 80 & -0.0858 & -0.1907 & -0.1679 & -0.1703 && -0.1604 & -0.1781 & -0.1773 & -0.1673 & -0.1731 \\
\midrule
0.0050 & 100 & -0.1032 & -0.1914 & -0.1679 & -0.1815 && -0.1757 & -0.1779 & -0.1780 & -0.1676 & -0.1690 \\
0.0125 & 100 & -0.1045 & -0.1740 & -0.1679 & -0.1791 && -0.1582 & -0.1779 & -0.1765 & -0.1676 & -0.1662 \\
0.0200 & 100 & -0.0932 & -0.1909 & -0.1680 & -0.1523 && -0.1601 & -0.1781 & -0.1773 & -0.1673 & -0.1592 \\
\bottomrule
\multicolumn{12}{p{0.99\linewidth}}{\small\rule{0pt}{3ex}\textbf{Notes}: (1) Average validation loss of end-to-end systems over four training folds. (2) The lowest validation score of each system is \textbf{bolded}.}\\
\end{tabular}
\label{table:validation}
\end{table*}

\begin{table*}[!ht]
\caption{Validation loss in Experiment 5}
\centering
\begin{tabular}{ccccr@{}ccr@{}cc}
\toprule
$\eta$ &  Epochs & \multicolumn{2}{c}{Linear} && \multicolumn{2}{c}{2-layer} && \multicolumn{2}{c}{3-layer}\\[0.5ex]\cline{3-4} \cline{6-7} \cline{9-10}
\rule{0pt}{3ex}&& Nom. & DR && Nom. & DR && Nom. & DR\\
\midrule
0.0050 & 20 & -0.2255 & -0.1167 && -0.1376 &  0.0350 && \textbf{-0.0603} & -0.1373 \\
0.0125 & 20 & -0.2391 & -0.0725 && -0.1966 & -0.0433 && -0.0601 & -0.1968 \\
0.0200 & 20 & -0.2677 & \textbf{-0.2475} && \textbf{-0.3014} & \textbf{-0.2481} && -0.0408 & \textbf{-0.2576} \\
\midrule
0.0050 & 40 & -0.2483 & -0.2237 && -0.1377 &  0.0348 && -0.0603 & -0.0645 \\
0.0125 & 40 & -0.2740 & -0.2270 && -0.1965 & -0.2419 && -0.0601 & -0.1758 \\
0.0200 & 40 & -0.2564 & -0.2140 && -0.2481 & -0.2116 && -0.0404 & -0.0423 \\
\midrule
0.0050 & 60 & \textbf{-0.2796} & -0.2346 && -0.1377 & -0.1143 && -0.0603 & -0.1177 \\
0.0125 & 60 & -0.2135 & -0.1729 && -0.1965 & -0.2088 && -0.0602 & -0.0832 \\
0.0200 & 60 & -0.2628 & -0.1796 && -0.2545 & -0.2115 && -0.0405 & -0.1353 \\
\bottomrule
\multicolumn{10}{p{0.7\linewidth}}{\small\rule{0pt}{3ex}\textbf{Notes}: (1) Validation loss of end-to-end systems over a single training fold. (2) The lowest validation score of each system is \textbf{bolded}.}\\
\end{tabular}
\label{table:exp5_val}
\end{table*}

\end{appendices}

\end{document}